\documentclass[10pt]{amsart}

\usepackage[OT2,T1]{fontenc}
\usepackage[utf8]{inputenc}
\usepackage[dvipsnames]{xcolor}
\colorlet{revcolor}{black}
\newcommand{\rev}[1]{\textcolor{revcolor}{#1}}

\usepackage{tikz-cd}
\usepackage{stmaryrd}
\SetSymbolFont{stmry}{bold}{U}{stmry}{m}{n}
\usepackage[pagebackref]{hyperref}
\hypersetup{
	pdfpagemode=FullScreen,
	colorlinks=true,
    linkcolor=purple,
    filecolor=blue,      
    urlcolor=blue,
    citecolor=ForestGreen,
    linktoc=page}

\usepackage{lmodern}
\usepackage{microtype}
\usepackage{amsfonts,mathrsfs}
\usepackage[abbrev]{amsrefs}
\usepackage[scr=rsfs]{mathalfa}
\usepackage{amsmath,amssymb,amsthm,mathtools,array,booktabs}
\usepackage{graphicx}
\usepackage{cancel}

\hypersetup{
  unicode,
  bookmarksnumbered,
  linktoc = all,
  pdfborderstyle = {/S/U/W 0.5}
}

\theoremstyle{plain}
\newtheorem{theorem}{Theorem}[section]
\newtheorem{theorema}{Theorem}

\newtheorem{proposition}[theorem]{Proposition}
\newtheorem{lemma}[theorem]{Lemma}

\newtheorem{conjecture}[theorem]{Conjecture}
\newtheorem*{conjecture*}{Conjecture}

\theoremstyle{definition}
\newtheorem{definition}[theorem]{Definition}

\newtheorem{instance}[theorem]{Example}
\newtheorem{remark}[theorem]{Remark}

\def\mC{\mathbb{C}}

\let\Re\relax
\let\Im\relax
\DeclareMathOperator{\Re}{Re}
\DeclareMathOperator{\Im}{Im}

\numberwithin{equation}{section}

\title[NC resolutions as mirror of singular CY varieties]{Non-commutative resolutions
as mirrors of singular Calabi--Yau varieties}
\author[Tsung-Ju Lee]{Tsung-Ju Lee}
\address{Department of Mathematics,National Cheng Kung University, Tainan 70101, Taiwan.}
\email{tsungju@gs.ncku.edu.tw}

\author[Bong H.~Lian]{Bong H.~Lian}
\address{Center for Mathematics and Interdisciplinary Sciences,
Fudan University, Shanghai, 200433, China}
\address{Shanghai Institute for Mathematics and Interdisciplinary Sciences, Yangpu District, Shanghai, 200433, China}
\email{lianbong@gmail.com}

\author[Mauricio Romo]{Mauricio Romo}
\address{Center for Mathematics and Interdisciplinary Sciences,
Fudan University, Shanghai, 200433, China}
\address{Shanghai Institute for Mathematics and Interdisciplinary Sciences, Yangpu District, Shanghai, 200433, China}
\email{mromoj@simis.cn}


\begin{document}
\begin{abstract}
It has been conjectured that the hemisphere partition function \cites{Honda:2013uca,Hori:2013ika} in a gauged linear sigma model (GLSM) computes the central charge \cite{MR2373143} of an object in the bounded derived category of coherent sheaves for Calabi--Yau (CY) manifolds. There is also evidence in \cites{Hosono:1995bm,Hosono:2000eb}.
On the other hand, non-commutative resolutions of singular CY varieties have been studied in the context of abelian GLSMs \cite{Caldararu:2010ljp}.
In this paper, we study an analogous construction of abelian GLSMs for non-commutative resolutions and propose that they can be used to study a class of recently discovered mirror pairs of singular CY varieties. Our main result shows that the hemisphere partition functions (a.k.a.~$A$-periods) in the new GLSM are in fact period integrals (a.k.a.~$B$-periods) of the singular CY varieties. We conjecture that the two are completely equivalent: $B$-periods are the same as $A$-periods. We give some examples to support this conjecture and formulate some expected 
homological mirror symmetry (HMS) relation between the GLSM theory and the CY. As shown in \cite{2020-Hosono-Lee-Lian-Yau-mirror-symmetry-for-double-cover-calabi-yau-varieties}, the $B$-periods in this case are precisely given by a certain fractional version of the $B$-series of \cite{Hosono:1995bm}. Since a hemisphere partition function is defined as a contour integral in a cone in the complexified secondary fan (or FI-theta parameter space) \cite{Hori:2013ika}, it can be reduced to a sum of residues \cites{zhdanov1998computation,passare1994multidimensional}. Our conjecture shows that this residue sum may now be amenable to computations in terms of the $B$-series.
\end{abstract}
\maketitle
\thispagestyle{empty}

\section{Introduction}

\subsection{Motivations} Homological mirror symmetry (HMS) \cite{kontsevich1995homological} for Calabi--Yau (CY) manifolds can be generally formulated in terms of bounded derived categories of coherent sheaves, and Fukaya categories. Thus much focus in recent decades has been to understand their relationships, and their consequences. However, not until recently, have people begun to ask similar questions about the case of singular varieties. By general considerations of marginal deformations of superconformal field theories \cite{Cecotti:1991me} one can argue that supersymmetric quantities (for example chiral/anti-chiral rings or supersymmetric boundary conditions) on a singular CY variety are equivalent to their counterparts on a crepant resolution (if one exists) of the CY variety. A mathematical counterpart of this is the so-called crepant resolution conjecture \cite{ruan2006cohomology}. Therefore, to formulate HMS for such a singular variety, it is natural to look for a categorical version of the resolution for the bounded derived category of coherent sheaves, and also a substitute for the Fukaya category in the singular case. 
One possibility for the former is to consider some kind of non-commutative (NC) resolution associated with a singular variety \cite{Caldararu:2010ljp}. 

This paper is an attempt to test the idea of using NC resolutions on a large class of recently discovered mirror pairs of singular CY varieties \cites{2020-Hosono-Lian-Takagi-Yau-k3-surfaces-from-configurations-of-six-lines-in-p2-and-mirror-symmetry-i,2020-Hosono-Lee-Lian-Yau-mirror-symmetry-for-double-cover-calabi-yau-varieties}. Let us recall briefly the results of \cites{2020-Hosono-Lian-Takagi-Yau-k3-surfaces-from-configurations-of-six-lines-in-p2-and-mirror-symmetry-i,2020-Hosono-Lee-Lian-Yau-mirror-symmetry-for-double-cover-calabi-yau-varieties}. It is shown that for each toric variety admitting a maximal projective crepant partial (MPCP) desingularization, equipped with a given nef-partition, one can construct a family of equisingular CY varieties as double covers on the toric variety. The branching locus of the double cover is given by an union of hypersurfaces in general position 
specified by the nef-partition. The period integrals of the double cover are formally certain fractional counterparts of period integrals of ordinary complete intersections (hence the term `fractional complete intersection’).
Most importantly, if two toric varieties equipped with nef-partitions are dual to each other, in the sense of Batyrev--Borisov \cite{1996-Batyrev-Borisov-on-calabi-yau-complete-intersections-in-toric-varieties}, then it has been shown that their corresponding double cover CY varieties are mirror to each other. In fact, we can apply many `mirror tests’ to see that this is in fact the case. For example, the Yukawa couplings of one family can be shown to compute the untwisted part of the 
genus zero orbifold Gromov--Witten invariants of the mirror family \cites{Lee:2025ab,Lee2026}. These new singular mirror pairs of CY therefore provide a very interesting testing ground for the idea \rev{of} non-commutative resolutions.
In this paper, we shall experiment with this idea from the point of view of abelian GLSMs.

\subsection{Statements of the main results}
We summarize our main results in this subsection. Given a 
reflexive polytope \(\Delta\), denote by \(\mathbf{P}_{\Delta}\)
the toric variety whose defining fan is the normal fan of \(\Delta\).
Recall that a nef-partition is a decomposition
of the anti-canonical divisor \(-K_{\mathbf{P}_{\Delta}}\)
into a sum of numerically effective Cartier divisors.
If \(X\to \mathbf{P}_{\Delta}\) is a MPCP desingularization,
then we obtain a nef-partition on \(X\), namely
\begin{eqnarray}
    -K_{X} = E_{1}+\cdots + E_{r},
\end{eqnarray}
by pullback. The nef-partition also determines a partition of
\(1\)-cones in the defining fan of \(X\). Let us list the 
\(1\)-cones in \(E_{i}\) as
\begin{eqnarray}
    \{\phi_{i,1},\ldots,\phi_{i,m_{i}}\}.
\end{eqnarray}
Assume that \(X\) is smooth. Recall that \(X\)
is a GIT quotient by an algebraic torus
\begin{eqnarray}
    X\cong \mathbb{C}^{p}\sslash_{\tau} (\mathbb{C}^{\ast})^{s}~\mbox{with}~p=
    m_{1}+\ldots+m_{r}
\end{eqnarray}
and \(\tau\) being a suitable character\footnote{See Section \ref{sec:GLSMmir} for details of the GIT construction.}.
We use the same symbols \(\phi=\phi_{1,1},\ldots,\phi_{r,m_{r}}\)
to denote the homogeneous coordinates attached to the corresponding 
\(1\)-cones. In the GIT representation, we denote by 
\(\theta^{k}_{i,j}\) the weight of \(\phi_{i,j}\) under the 
\(k\)\textsuperscript{th} component in 
\((\mathbb{C}^{\ast})^{s}\).

As we have explained, we are interested in CY double covers of
\(X\) whose branching locus is the union of all toric divisors
as well as the zero loci of general sections \(f_{1}(\phi),\ldots,f_{r}(\phi)\) 
(in homogeneous coordinates)
of \(E_{1},\ldots,E_{r}\). By deforming \(f_{i}(\phi)\) and leaving all
toric divisors unchanged,
we obtain a family of double covers \(\mathcal{Y}\to U\).
In the classical case of CY complete intersections in \(X\), one
considers a superpotential of the form
\begin{eqnarray}
    \sum_{i=1}^{r} p_{i} f_{i}(\phi)
\end{eqnarray}
where \(p_{i}\) is the
coordinate of the line bundle \(\mathcal{O}(-E_{i})\) along
the fiber.

In this paper, to investigate 
CY double covers of \(X\), we introduce a \emph{curved superpotential}
(cf.~Definition \ref{def:curved-superpot})
\begin{equation}
\label{eq:curve-superpot-intro}
    W(\phi,z):=\sum_{i,j} z_{i,j}^{2}\phi_{i,j} + 
    \sum_{i=1}^{r} z_{i,0}^{2}f_{i}(\phi).
\end{equation}
Note that \((\prod_{i,j}\phi_{i,j})\prod_{i=1}^{r} f_{i}(\phi)=0\)
is the branching divisor. Here comes the crucial point. 
In order to make \eqref{eq:curve-superpot-intro}
to be \((\mathbb{C}^{\ast})^{s}\)-invariant, we
assign the weights as follows.
\begin{itemize}
    \item We double the weight of each \(\phi_{i,j}\), so it has weight \((2\theta_{i,j}^{1},\ldots,2\theta_{i,j}^{s})\) now.
    \item Let \(z_{i,j}\) carry the weight \((-\theta_{i,j}^{1},\ldots,-\theta_{i,j}^{s})\).
    \item Let \(z_{i,0}\) carry the weight \((-\sum_{j=1}^{m_{i}}\theta_{i,j}^{1},\ldots,-\sum_{j=1}^{m_{i}}\theta_{i,j}^{s})\).
\end{itemize}
Under the assignment, the function \eqref{eq:curve-superpot-intro} becomes
a \((\mathbb{C}^{\ast})^{s}\)-invariant function 
\begin{eqnarray}
    W\colon \mathbb{C}^{p}\times \mathbb{C}^{p+r}\to \mathbb{C}.
\end{eqnarray}
Thus we obtain a \emph{GLSM for CY double covers}.

Our main result relates the \(A\)-periods (a.k.a.~hemisphere 
partition functions) for this GLSM and
the \(B\)-periods for the proposed mirror of the double cover family
\(\mathcal{Y}\to U\). In \cite{2020-Hosono-Lee-Lian-Yau-mirror-symmetry-for-double-cover-calabi-yau-varieties}, a hypothetical 
mirror for the double cover family was constructed and 
it was proven that they form a topological mirror pair when the 
dimension is less than \(5\).
Like the classical case, the \(B\)-periods here
can be also studied by GKZ systems. It can be shown that the \(B\)-periods
indicated above are governed by a certain GKZ system with 
fractional parameter.
Nevertheless, to compute the \(A\)-periods, one
needs to investigate the brane factors in the
hemisphere partition function. This involves 
constructing various 
matrix factorizations for the curved superpotential \(W\).
Fortunately, owing to the combinatorial nature in the toric regime, for each subset 
\(J\) of the set of \(1\)-cones in the fan defining \(X\), we are able
to construct a matrix factorization of \(W\) associated to the subset \(J\)
and explicitly compute the brane factor \(\mathfrak{B}_{J}\).
Consequently we arrive at an explicit formula for the corresponding
\(A\)-period \(Z_{\mathfrak{B}_{J}}\).
Our main result is
\begin{theorema}[=Theorem \ref{thm:main-theorem}]
Under a suitable change of variable, the \(A\)-periods \(\hat{Z}_{\mathfrak{B}_{J}}\) 
satisfy the GKZ system for the \(B\)-periods for the proposed mirror 
of the double cover family \(\mathcal{Y}\to U\).
Here \(\hat{Z}_{\mathfrak{B}_{J}}\) is the \(A\)-period
obtained by performing the change of 
variable on \(Z_{\mathfrak{B}_{J}}\).
\end{theorema}

We also conjecture that when \(J\) runs through all
subsets, one obtains a generating set of the solution space to 
the GKZ system, hence all the \(A\)-periods for the GLSM.
To support the conjecture, we investigate the case \(X=\mathbf{P}^{n}\)
and prove the following result.
\begin{theorema}[=Theorem \ref{thm:conjecture-pn}]
When \(J\) runs through all subsets of the set of \(1\)-cones
in the defining fan of \(\mathbf{P}^{n}\), the \(A\)-periods
\(\hat{Z}_{\mathfrak{B}_{J}}\) generate
the full solution space of the GKZ system.
Moreover, if we denote by 
\begin{equation}
    \{\phi_{1},\ldots,\phi_{n+1}\}
\end{equation}
the set of \(1\)-cones in the fan defining \(\mathbf{P}^{n}\), then
\begin{equation}
    \{\hat{Z}_{\mathfrak{B}_{J_{m}}}\bigm\vert J_{m}=
    \{\phi_{1},\ldots,\phi_{m}\}~\mbox{and}~0\le m\le n\}
\end{equation}
where \(J_{0}:=\emptyset\)
is a basis for the solution space.
\end{theorema}




The paper is organized as follows. 

In Section \ref{sec:SingCY}, we give an overview of the construction of mirror pairs of singular CY varieties, given by double covers of toric varieties. This general construction was introduced in \cite{2020-Hosono-Lee-Lian-Yau-mirror-symmetry-for-double-cover-calabi-yau-varieties}. We recall some basics in toric geometry, and outline the construction of double covers on a general toric variety admitting an MPCP 
desingularization, equipped with a nef-partition. We then summarize some of the main results in \cite{2020-Hosono-Lee-Lian-Yau-mirror-symmetry-for-double-cover-calabi-yau-varieties} on singular CY mirror pairs, and illustrate them in examples. 

In Section \ref{sec:GLSMmir}, we introduce the construction of the abelian GLSMs, which will later on play the role of NC resolutions for singular CYs given by double covers. We start by reviewing the setup for two dimensional $\mathcal{N}=(2,2)$ gauge theories with boundaries, and give an overview of GLSMs. We spell out the assumptions we impose on the GLSMs to be considered in this paper. In particular, we consider nonanomalous GLSMs, whose IR theory is completely characterized by the classical Higgs branch. The GLSMs of interest will be hybrid models: their underlying target space is an orbifold vector bundle over a base given by the critical locus of a superpotential function. We consider a particular form of such a `curved’ superpotential
and propose that the resulting GLSM is a quantum field theory realization of a NC resolution for our double cover singular CY.
Key ingredients introduced here include the matrix factorization category for a given gauge group and a superpotential, and the notion of the hemisphere partition function of an object in this category. Most of the expository discussion here follows \cites{Honda:2013uca,Hori:2013ika}. It has been conjectured that for smooth CYs, the hemisphere partition function coincides with the so-called the $A$-period or the central charge of an object. The latter was introduced and studied mathematically by \cites{MR2373143,Hosono:2000eb,Iritani:2009ab} from various viewpoints. Finally, we also review some basics on toric GITs, and the relation between the secondary fan and the (stringy K\"ahler) moduli space of CYs in mirror symmetry (or the FI-theta parameter space in GLSM theory).

In Section \ref{sec:Periods}, we formulate and prove our main result in Theorem \ref{thm:main-theorem}. For each toric variety $X$ equipped with a choice of nef-partition, we consider the $A$-periods of the GLSM realizing the (singular) CY double covers $Y$ of $X$. In \cite{2020-Hosono-Lee-Lian-Yau-mirror-symmetry-for-double-cover-calabi-yau-varieties} it is shown that the sheaf of $B$-periods of $Y^{\vee}$ (the mirror of $Y$) can be completely characterized by a GKZ system. 
We then show that the $A$-periods are solutions to the same GKZ system, hence proving the $A$-periods of $Y$ are in fact $B$-periods $Y^{\vee}$. A crucial step here is the explicit determination of the $B$-brane factor in the hemisphere partition function in this case.
We end with a discussion of our conjecture on the equivalence of the two kinds of periods, and provide some numerical evidence.

\noindent{\bf Acknowledgments}.~The authors thank W.~Gu for collaboration at an early stage of this work. We would like to thank L.~Borisov, D.~Pomerleano, and E.~Scheidegger
for discussions and comments. We would like to thank S.~Hosono for 
his collaboration, which helped inspire this project.
We would like to thank T.~Pantev for his interest in this project. MR thanks Harvard CMSA, Rutgers University, Heidelberg University and Uppsala University for hospitality while
part of this work has been performed. BHL would like to thank YMSC and BIMSA, where part of this collaboration was done.

MR acknowledges support from the National Key Research and Development Program of China, 
grant No.~2020YFA0713000, the Research Fund for International
Young Scientists, NSFC grant No.~1195041050. TJL is supported by 
NSTC grant 112-2115-M-006-016-MY3
and was partially supported by AMS--Simons travel grant. 
BHL is partially supported by the Simons Collaboration Grant on 
Homological Mirror Symmetry and Applications 2015-2023.

\section{\label{sec:SingCY} Calabi--Yau double covers and their mirrors}
In this section, we recall the construction 
of pairs of singular Calabi--Yau double covers $(Y,Y^{\vee})$ 
in \cite{2020-Hosono-Lee-Lian-Yau-mirror-symmetry-for-double-cover-calabi-yau-varieties}
and review their properties.
To this end, let us fix the notation that will be 
used throughout this note.
\begin{enumerate}
\item[(1)] Let \(N=\mathbb{Z}^{n}\) be a rank \(n\) lattice and
\(M=\mathrm{Hom}_{\mathbb{Z}}(N,\mathbb{Z})\) be its dual lattice. 
Let \(N_{\mathbb{R}}:=N\otimes_{\mathbb{Z}}\mathbb{R}\) and 
\(M_{\mathbb{R}}:=M\otimes_{\mathbb{Z}}\mathbb{R}\).
We denote by \(\langle-,-\rangle\) the canonical
dual pairing between \(M\) and \(N\).
\item[(2)] Let \(\Sigma\) be a fan in \(N_{\mathbb{R}}\).
We denote by \(\Sigma(k)\) the set of \(k\)-dimensional cones in \(\Sigma\).
In particular, \(\Sigma(1)\) is the set of \(1\)-cones in \(\Sigma\). 
Similarly, for a cone \(\sigma\in\Sigma\),
we denote by \(\sigma(1)\) the set of \(1\)-cones belonging to \(\sigma\).
By abuse of the notation, we also denote by \(\rho\) 
the primitive generator of the corresponding 
\(1\)-cone.
\item[(3)] Denote by \(X_{\Sigma}\)
the toric variety determined by the fan \(\Sigma\). 
Each \(\rho\in\Sigma(1)\) determines a torus-invariant Weil divisor 
\(D_{\rho}\) on \(X_{\Sigma}\).
Any torus-invariant Weil divisor \(D\) is linearly equivalent to
\(\sum_{\rho\in\Sigma(1)} a_{\rho}D_{\rho}\). We 
define the \emph{polyhedron} of \(D\) 
\begin{equation*}
\Delta_{D}:=\left\{m\in M_{\mathbb{R}}~|~\langle m,\rho\rangle\ge -a_{\rho}~
\mbox{for all}~\rho\right\}.
\end{equation*}
\rev{Note that \(\Delta_{D}\) is a polytope if \(\Sigma\) is a complete fan and \(\Delta_{D}\) is called the \emph{polytope} of \(D\)}.
The integral points \(M\cap\Delta_{D}\) gives rise to a canonical
basis of \(\mathrm{H}^{0}(X_{\Sigma},D)\).
\item[(4)] A polytope in \(M_{\mathbb{R}}\) is called a \emph{lattice polytope}
if its vertices belong to \(M\). For a lattice polytope \(\Delta\)
in \(M_{\mathbb{R}}\), we denote by \(\Sigma_{\Delta}\) the normal fan of 
\(\Delta\). The toric variety determined by \(\Delta\) is denoted by \(\mathbf{P}_{\Delta}\),
i.e., \(\mathbf{P}_{\Delta}=X_{\Sigma_{\Delta}}\).
\item[(5)] A \emph{reflexive polytope} \(\Delta\subset M_{\mathbb{R}}\) is a lattice polytope 
which contains the origin \(\mathbf{0}\in M_{\mathbb{R}}\) in its 
interior and such that the polar dual 
\begin{equation*}
\Delta^{\vee}:=\{n\in N_{\mathbb{R}}~|~\langle m,n\rangle\ge -1~\mbox{for}~m\in \Delta\}
\end{equation*}
is again a lattice polytope. If
\(\Delta\) is a reflexive polytope, then \(\Delta^{\vee}\) is also a lattice
polytope and satisfies \((\Delta^{\vee})^{\vee}=\Delta\). The normal fan of \(\Delta\) (resp.~face fan of \(\Delta\))
is the face fan of \(\Delta^{\vee}\) (resp.~the normal fan of \(\Delta^{\vee}\)).
\end{enumerate}

\subsection{The Batyrev--Borisov's duality construction}
\label{subsection:b-b-construction}
Let us begin with the notion of nef-partitions.
Let \(\Delta\subset M_{\mathbb{R}}\) be a reflexive polytope.
Recall that a \emph{nef-partition} on \(\mathbf{P}_{\Delta}\) 
is a decomposition of \(\Sigma_{\Delta}(1)=\sqcup_{k=1}^{r} I_{k}\)
such that each \(E_{k}:=\sum_{\rho\in I_{k}} D_{\rho}\) is 
numerically effective, i.e., \(D.C\ge 0\) for any irreducible
complete curve \(C\subset \mathbf{P}_{\Delta}\). 
Note that \(E_{1}+\cdots+E_{r}=-K_{\mathbf{P}_{\Delta}}\). 
This also gives rise to a Minkowski sum decomposition
\begin{equation*}
\Delta = \Delta_{1}+\cdots+\Delta_{r}~\mbox{where}~\Delta_{i}:=\Delta_{E_{i}}.
\end{equation*}
By abuse of terminology, both \(E_{1}+\cdots+E_{r}=-K_{\mathbf{P}_{\Delta}}\)
and \(\Delta = \Delta_{1}+\cdots+\Delta_{r}\) are called nef-partitions.

Let \(I_{1},\ldots,I_{r}\) be a nef-partition on \(\mathbf{P}_{\Delta}\).
Denote 
\begin{equation*}
\nabla_{k}=\operatorname{Conv}(I_{k}\cup\mathbf{0})~\mbox{and}~
\nabla = \nabla_{1}+\cdots+\nabla_{r}.
\end{equation*}
Borisov proved that \(\nabla\) is a reflexive polytope in \(N_{\mathbb{R}}\)
whose polar dual is 
\begin{equation*}
    \nabla^{\vee}=\mathrm{Conv}(\Delta_{1},\ldots,\Delta_{r})
\end{equation*}
and \(\nabla_1+\cdots+\nabla_{r}\) corresponds to a nef-partition on \(\mathbf{P}_{\nabla}\)
\cite{1993-Borisov-towards-the-mirror-symmetry-for-calabi-yau-complete-intersections-in-gorenstein-toric-fano-varieties}. This is 
called the \emph{dual nef-partition}
in \cite{1996-Batyrev-Borisov-on-calabi-yau-complete-intersections-in-toric-varieties}.
The corresponding nef toric divisors are denoted by \(F_{1},\ldots,F_{r}\).
Then the polytope of \(F_{j}\) is \(\nabla_{j}\).

Let \(X\to \mathbf{P}_{\Delta}\) and \(X^{\vee}\to \mathbf{P}_{\nabla}\)
be maximal projective crepant partial (MPCP for short hereafter) desingularization
for \(\mathbf{P}_{\Delta}\) and \(\mathbf{P}_{\nabla}\).
Recall that the polytopes \(\Delta_{i}\) and \(\nabla_{j}\)
correspond to \(E_{i}\) on \(\mathbf{P}_{\Delta}\) and 
\(F_{j}\) on \(\mathbf{P}_{\nabla}\).
The nef-partitions on \(\mathbf{P}_{\Delta}\) and \(\mathbf{P}_{\nabla}\)
pullback to nef-partitions on \(X\) and \(X^{\vee}\). 
To save the notation, the corresponding nef-partitions and toric divisors 
on \(X\) and \(X^{\vee}\) will be still denoted by \(\Delta_{i}\), \(\nabla_{j}\) and
\(E_{i}\), \(F_{j}\) respectively.

\subsection{Calabi--Yau double covers}
\label{subsection:cy-double-covers}
Suppose we are given the data in \S\ref{subsection:b-b-construction} and
let notation be the same as there.
Throughout this note, unless otherwise stated, we assume that 
\begin{center}
\textit{Both \(\mathbf{P}_{\Delta}\) and \(\mathbf{P}_{\nabla}\) admit a smooth MPCP
desingularization},
\end{center}
i.e., we assume that both \(\Delta\) and \(\nabla\) admit 
fine regular star triangulations (FRSTs); each simplex
in the triangulation is uni-modular (a full-dimensional lattice polytope 
of minimal possible Euclidean volume) and has \(\mathbf{0}\) as its vertex.

From the duality, we have
\begin{equation*}
\mathrm{H}^{0}(X^{\vee},F_{i})\cong 
\bigoplus_{\rho\in\nabla_{i}\cap N}\mathbb{C}\cdot t^{\rho}~\mbox{and}~
\mathrm{H}^{0}(X,E_{i})\cong 
\bigoplus_{m\in\Delta_{i}\cap M}\mathbb{C}\cdot t^{m}.
\end{equation*}
Here we use the same notation \(t=(t_{1},\ldots,t_{n})\) to 
denote the coordinates on the maximal torus of \(X^{\vee}\) and \(X\).

A double cover \(Y^{\vee}\to X^{\vee}\) has trivial 
canonical bundle if and only if 
the branching locus is linearly equivalent to \(-2K_{X^{\vee}}\).
Let \(Y^{\vee}\to X^{\vee}\) be the double cover
constructed from the section \(s=s_{1}\cdots s_{r}\) with
\begin{equation*}
(s_{1},\ldots,s_{r})\in \mathrm{H}^{0}(X^{\vee},2F_{1})\times\cdots\times
\mathrm{H}^{0}(X^{\vee},2F_{r}).
\end{equation*}
We assume that \(s_{i}\in\mathrm{H}^{0}(X^{\vee},2F_{i})\) is 
of the form \(s_{i}=s_{i,1}s_{i,2}\) with \(s_{i,1},s_{i,2}\in \mathrm{H}^{0}(X^{\vee},F_{i})\)
and moreover that \(s_{i,1}\) is the section corresponding to the lattice point
\(\mathbf{0}\in\nabla_{i}\cap N\)
and that \(\operatorname{div}(s)\) is a divisor with strictly normal crossings. 
This procedure, which is inspired by the work 
\cites{
2020-Hosono-Lian-Takagi-Yau-k3-surfaces-from-configurations-of-six-lines-in-p2-a
nd-mirror-symmetry-i,
2019-Hosono-Lian-Yau-k3-surfaces-from-configurations-of-six-lines-in-p2-and-mirr
or-symmetry-ii-lambda-k3-functions}, is called the \emph{partial gauge fixing}.
\begin{remark}
Under the assumption that 
\(s_{i,1}\) is the section corresponding to the lattice point
\(\mathbf{0}\in\nabla_{i}\cap N\), when deforming \(s_{i,2}\), the period integrals for the 
resulting family will satisfy a certain type of GKZ system. This allows us to describe 
period integrals using GKZ systems.
\end{remark}
Consequently, we obtain a subfamily of double covers of \(X^{\vee}\) by deforming \(s_{i,2}\)
The family is parameterized by an open subset
\begin{equation*}
V\subset \mathrm{H}^{0}(X^{\vee},F_{1})
\times\cdots\times\mathrm{H}^{0}(X^{\vee},F_{r})
\end{equation*}
\rev{where \(V\) consist of \(r\)-tuples \((s_{1,2},\ldots,s_{r,2})\)
such that 
\begin{equation}
\label{eq:div-cond}
    \bigcup_{i=1}^{r} \operatorname{div}(s_{i,1})\cup 
    \operatorname{div}(s_{i,2})
\end{equation}
is a simple normal crossing divisor.}

\begin{definition}
\label{definition:gauged-fixed-double-cover-family}
Given a decomposition \(\nabla=\nabla_{1}+\cdots+\nabla_{r}\) representing
a nef-partition \(F_{1}+\cdots+F_{r}\) on \(X^{\vee}\) as above,
the subfamily \(\mathcal{Y}^{\vee}\to V\) constructed above is called the 
\emph{gauge fixed double cover family branched along
the nef-partition over \(X^{\vee}\)} or simply the \emph{gauge fixed double cover family} if 
no confusion occurs.
\end{definition}

Likewise, applying the construction to the
decomposition \(\Delta=\Delta_{1}+\cdots+\Delta_{r}\) representing 
the dual nef-partition \(E_{1}+\cdots+E_{r}\) 
on \(X\) yields
another family \(\mathcal{Y}\to U\),
where \(U\) is an open subset in 
\begin{equation*}
\mathrm{H}^{0}(X,E_{1})\times\cdots\times\mathrm{H}^{0}(X,E_{r})
\end{equation*} 
\rev{satisfying a similar condition as in \eqref{eq:div-cond}.}
Denote by \(Y\) (resp.~\(Y^{\vee}\))
the fiber of \(\mathcal{Y}\to U\) (resp.~\(\mathcal{Y}^{\vee}\to V\)).
It is conjectured that
\begin{conjecture*}[\cite{2020-Hosono-Lee-Lian-Yau-mirror-symmetry-for-double-cover-calabi-yau-varieties}]
\(Y\) and \(Y^{\vee}\) are mirror.
\end{conjecture*}

\begin{remark}
\label{rmk:branching-locus}
    The branching locus of \(Y\) is given by \(\operatorname{div}(s)=0\);
    this is a strictly normal crossing divisors whose components consist of
    all toric divisors of \(X\) and \(r\) smooth sections of the bundles \(E_{i}\), \(i=1,\ldots,r\).
\end{remark}

The following proposition is also proven in 
\cite{2020-Hosono-Lee-Lian-Yau-mirror-symmetry-for-double-cover-calabi-yau-varieties}
and provides some numerical evidence of the conjecture.
Recall that \(Y\) (likewise \(Y^{\vee}\)) is an orbifold, the singular cohomology 
group \(\mathrm{H}^{p}(Y,\mathbb{C})\) admits a pure Hodge structure.

\begin{proposition}
\label{prop:hodge-numbers}
We have \(\chi_{\mathrm{top}}(Y)=(-1)^{n}\chi_{\mathrm{top}}(Y^{\vee})\).
Here \(\chi_{\mathrm{top}}(-)\) is the topological Euler characteristic.
Moreover, we have for \(p+q\ne n\), 
\begin{equation*}
h^{p,q}(X)=h^{p,q}(Y)~\mbox{and}~h^{p,q}(X^{\vee})=h^{p,q}(Y^{\vee}).
\end{equation*}
Consequently, when \(n\le 4\), we have \(h^{p,q}(Y)=h^{n-p,q}(Y^{\vee})\), i.e.,
\(Y\) and \(Y^{\vee}\) form a topological mirror pair.
\end{proposition}
Indeed, in \cite{2020-Hosono-Lee-Lian-Yau-mirror-symmetry-for-double-cover-calabi-yau-varieties},
it is shown that 
\begin{equation*}
\chi_{\mathrm{top}}(Y) = \chi_{\mathrm{top}}(X) + (-1)^{n} \chi_{\mathrm{top}}(X^{\vee}).
\end{equation*}
It follows that 
\begin{equation*}
\dim\mathrm{H}^{n}(Y,\mathbb{C})=\dim\mathrm{H}^{n}(X,\mathbb{C})+
\chi_{\mathrm{top}}(X^{\vee}).
\end{equation*}

\begin{instance}[Calabi--Yau double cover of \(\mathbf{P}^{3}\) branched over eight hyperplanes]
\label{ex:p3}
Consider the reflexive polytope 
\begin{equation}
    \Delta = \operatorname{Conv}(\{(3,-1,-1),(-1,3,-1),(-1,-1,3),(-1,-1,-1)\}).
\end{equation}
We then have \(\mathbf{P}_{\Delta}\cong\mathbf{P}^{3}\). In the present case, 
we have \(X=\mathbf{P}_{\Delta}\). (Recall that \(X\) is a MPCP desingularization.)

Let us consider the nef-partition \(E_{1}+E_{2}+E_{3}+E_{4}=H+H+H+H=-K_{X}\) on \(X\). 
Here \(H\) is the hyperplane class of \(X\) and the partition corresponds to 
the partition on the set of \(1\)-cones
\begin{equation}
\{(1,0,0)\}\cup\{(0,1,0)\}\cup\{(0,0,1)\}\cup\{(-1,-1,-1)\}.
\end{equation}
Then the associated gauge fixed double cover family
\(\mathcal{Y}\to U\) is the family of double covers of \(\mathbf{P}^{3}\) branched along eight hyperplanes in
general position. Let \(Y\) denote a fiber in the family. The Hodge diamond of \(Y\) is given by
\begin{equation}
\label{eq:hodge-diamond-p3}
\begin{tikzcd}[column sep=0.1em, row sep=0.1em]
& & & &1 & & &\\
& & &0 & &0 & &\\
& & 0& &1 & &0 &\\
&1 & &9 & &9 & &1\\
& & 0& &1 & &0 &\\
& & &0 & &0 & &\\
& & & &1 & & &
\end{tikzcd}
\end{equation}
Let us investigate the mirror. From Batyrev--Borisov's duality construction,
the dual polytope associated with the partition \(H+H+H+H=-K_{X}\)is
\begin{equation} 
\nabla = \nabla_{1}+\nabla_{2}+\nabla_{3}+\nabla_{4}
\end{equation}
where \(\nabla_{i}=\operatorname{Conv}(\{\mathbf{0}\}\cup (\delta_{i1},\delta_{i2},\delta_{i3}))\) for \(i=1,2,3\) and
\(\nabla_{4}=\operatorname{Conv}(\{\mathbf{0}\}\cup (-1,-1,-1))\). (\(\nabla\) is a zonotope.)
In the present case, \(\mathbf{P}_{\nabla}\) is singular and admits a MPCP desingularization \(X^{\vee}\to\mathbf{P}_{\nabla}\).
Denote by \(\mathcal{Y}^{\vee}\to V\) the gauge fixed double cover family over \(X^{\vee}\) branched along
the dual nef-partition \(F_{1}+F_{2}+F_{3}+F_{4}=-K_{X^{\vee}}\) and let \(Y^{\vee}\) be a fiber.
One can check that the Hodge diamond of \(Y^{\vee}\) is given by
\begin{equation}
\label{eq:hodge-diamond-dual-p3}
\begin{tikzcd}[column sep=0.1em, row sep=0.1em]
& & & &1 & & &\\
& & &0 & &0 & &\\
& & 0& &9 & &0 &\\
&1 & &1 & &1 & &1\\
& & 0& &9 & &0 &\\
& & &0 & &0 & &\\
& & & &1 & & &
\end{tikzcd}
\end{equation}
It is easy to see that \eqref{eq:hodge-diamond-p3} and 
\eqref{eq:hodge-diamond-dual-p3} are mirror Hodge diamonds.
\end{instance}

\begin{remark}
    For \(X=\mathbf{P}^{3}\) different nef-partitions
    give rise to inequivalent Calabi--Yau double covers
    (cf.~\cite{Sawomir-Tomasz-double-covers}).
\end{remark}

\begin{remark}
We expect that
the Hodge diamonds of \(Y\) and \(Y^{\vee}\) for general smooth toric bases are related in
a simple way; they are identical after a \(\pi\slash 2\)-rotation
(cf.~Figure~\ref{fig:hodge}).
\begin{figure}
\begin{center}
\includegraphics[scale=1.2]{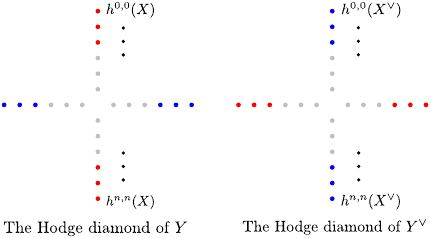}
\caption{The picture shows the expected Hodge diamonds of \(Y\) and \(Y^{\vee}\); the red dots denote
\(h^{p,p}(X)\) for \(0\le p\le n\) and the blue ones
denote \(h^{q,q}(X^{\vee})\) for \(0\le q\le n\).
}
\end{center}
\label{fig:hodge}
\end{figure}

\end{remark}

\section{\label{sec:GLSMmir} GLSMs as NC resolutions of \texorpdfstring{\(Y\)}{Y}}


\subsection{2d \texorpdfstring{$\mathcal{N}=(2,2)$}{} Gauge Theories with Boundaries}

In this section we will review certain aspects of gauged linear sigma 
models (GLSM) \cite{Witten:1993yc} with boundaries \cite{Herbst:2008jq}. We 
will 
define a set that we will denote GLSM data by the following elements:
\begin{itemize}
  \item \textbf{Gauge group}: a compact Lie group $\mathsf{G}$.
  \item \textbf{Chiral matter fields}: a faithful unitary representation 
$\rho_{m}\colon G\rightarrow \mathrm{GL}(V)$ of $\mathsf{G}$ on some complex vector space $V\cong 
\mathbb{C}^{N}$.
  \item \textbf{Superpotential}: a holomorphic, $G$-invariant polynomial $W\colon
V\rightarrow \mathbb{C}$, namely $W\in \mathrm{Sym}(V^{\vee})^\mathsf{G}$.
  \item \textbf{Fayet--Iliopolous (FI)-theta parameters}: a set of complex 
parameters $t$ such that
  \begin{eqnarray}
  \exp(t)\in \mathrm{Hom}(\pi_{1}(  \mathsf{G} ),\mathbb{C}^{*})^{\pi_{0}(\mathsf{G})}
  \end{eqnarray}
  i.e., $\exp(t)$ is a group homomorphism from $\pi_{1}(  \mathsf{G} )$ to $\mathbb{C}^{*}$ 
that is invariant under the adjoint action of $  \mathsf{G} $\footnote{Recall that $\pi_{0}(  \mathsf{G} )\cong   \mathsf{G} /  \mathsf{G} _{0}$, where $  \mathsf{G} _{0}$ is the identity component of $\mathsf{G}$, is the only subset of $\mathsf{G}$ that acts nontrivially on $\pi_{1}(  \mathsf{G} )$.}. It is customary to write 
$t=\zeta-\mathrm{i}\,\theta$, therefore \cite{hori2019notes}
  \begin{eqnarray}
  t\in \left(\frac{\mathfrak{t}^{\vee}_{\mathbb{C}}}{2\pi \mathrm{i}\,
\mathrm{P}}\right)^{W_{  \mathsf{G} }}\cong\frac{\mathfrak{z}^{\vee}_{\mathbb{C}}}{2\pi \mathrm{i}\,
\mathrm{P}^{W_{  \mathsf{G} }}},
  \end{eqnarray}
  where $\mathrm{P}$ is the weight lattice, $W_{  \mathsf{G} }$ is the Weyl subgroup of 
$  \mathsf{G} $, 
$\mathfrak{t}=\mathrm{Lie}(T)$ is the Cartan subalgebra of $\mathfrak{g}=\mathrm{Lie}(  \mathsf{G} )$ and 
$\mathfrak{z}=\mathrm{Lie}(Z(  \mathsf{G} ))$. 
  
\rev{  \item \textbf{\(R\)-symmetry}: a vector
\(R\)-symmetry, which will come in the form of a $\mathbb{C}^{*}$ action on $V$ that commute with the action of $\mathsf{G}$. This action is determined by a real representation of the Lie algebra $\mathbf{r}:\mathfrak{u}(1)\rightarrow 
\mathfrak{gl}(V)$. We will denote $R:\mathbb{C}^{*}\rightarrow \mathrm{GL}(V)$ the exponentiation of $\mathbf{r}$:
 \begin{eqnarray}\label{vectorRcharge}
R(\lambda)=\lambda^{\mathbf{r}},\qquad \lambda\in\mathbb{C}^{*}.
\end{eqnarray}
    The superpotential $W$ is required to
have weight $2$ under the vector R-symmetry:
    \begin{eqnarray}
    W(R(\lambda)\cdot\phi)=\lambda^{2}W(\phi), \qquad \text{for all \ }\lambda\in\mathbb{C}^{*}
    \end{eqnarray}
    where $\phi$ denotes the coordinates in $V$.}
  \end{itemize}
\begin{definition}
\label{definition:glsm-data}
We call a tuple $(  \mathsf{G} ,W,\rho_{m},t,R)$ satisfying the conditions 
above a GLSM data.
We call $(  \mathsf{G} ,W,\rho_{m},t,R)$ a \emph{nonanomalous} GLSM data if
furthermore the representation $\rho_{m}$ factors through $\mathrm{SL}(V)$.
\end{definition}  

\begin{remark}
In the following, since we are interested in properties of CY varieties as complex geometries, we will find it convenient to use the complexified gauge group. We use the short notation
$$
G:=\mathsf{G}_{\mathbb{C}}
$$
\end{remark}
  
In the following we will only work with nonanomalous GLSM data which is the 
relevant case for CY mirrors. We will focus our attention on B-type boundary 
conditions, 
that is, boundary conditions preserving the combination of supercharges 
$\mathbf{Q}_{B}:=\overline{\mathbf{Q}}_{+}+\overline{\mathbf{Q}}_{-}$ (and its 
charge conjugate $\mathbf{Q}^{\dag}_{B}:=\mathbf{Q}_{+}+\mathbf{Q}_{-}$) 
\cites{Herbst:2008jq,Hori:2013ika}. These 
boundary conditions are termed B-branes and they will play a central role in 
the 
present work. We define them for a fixed value of the FI-theta parameter 
$t$ and superpotential $W$. Moreover they form a triangulated 
category. We denote such a category by $\mathrm{MF}_{  \mathsf{G} }(W)$. \rev{Its objects will be denoted by
$(\mathcal{B},[L_{c}])$, where $[L_{c}]$ is some auxiliary data (defined below), that is used to keep track of the chamber $C$ where $\zeta=\mathrm{Re}(t)$ belongs. Without the inclusion of this auxiliary data, the category  $\mathrm{MF}_{  \mathsf{G} }(W)$ is just a special case of the more general category of coherent matrix factorizations defined in \cite{ballard2019variation}}. Let us start 
defining $\mathcal{B}$, the \emph{algebraic data}:
\begin{definition}
\label{definition:algebraic-data}
The algebraic data of the element \rev{$(\mathcal{B},[L_{c}])\in \mathrm{MF}_{G}(W)$} is a 
quadruple $\mathcal{B}=(M,\rho_{M},R_{M},\mathbf{T})$ where the elements are 
defined as:
\begin{itemize}
\item \textbf{Chan--Paton vector space}: a $\mathbb{Z}_{2}$-graded, finite
rank free $\mathrm{Sym}(V^{\vee})$-module denoted by $M=M_{0}\oplus 
M_{1}$.
\item \rev{\textbf{Boundary gauge and (vector) $R$-symmetry representation}: $\rho_{M}\colon 
G\rightarrow \mathrm{GL}(M)$, and $R_{M}\colon \mathbb{C}^{*}\rightarrow \mathrm{GL}(M)$ (defined likewise \eqref{vectorRcharge}) commuting and even 
representations.}
\item \textbf{Matrix factorization of $W$}: Also known as the \emph{tachyon profile}, 
a 
$\mathbb{Z}_2$-odd endomorphism $\mathbf{T} \in 
\mathrm{End}^{1}_{\mathrm{Sym}(V^{\vee})}(M)$ satisfying 
$\mathbf{T}^{2}=W\cdot\mathrm{id}_{M}$.
\end{itemize}
The group actions $\rho_{M}$ and $R_{M}$ must be compatible with $\rho_{m}$ and 
$R$, i.e.,
  for all $\lambda\in U(1)_{V}$ and $g\in   \mathsf{G} $, we demand
    \begin{equation}\label{rhodef}
   \begin{aligned}
      R_{M}(\lambda)\mathbf{T}(R(\lambda)\phi)R_{M}(\lambda)^{-1} & = \lambda 
\mathbf{T}(\phi) , \\
\rho_{M}(g)^{-1}\mathbf{T}(\rho_{m}(g)\cdot \phi)\rho_{M}(g) & = 
\mathbf{T}(\phi) .
    \end{aligned}
  \end{equation}
\end{definition}   

For later use, we denote the weights of $\rho_{m}$ as 
$Q_{j}\colon\mathfrak{t}\rightarrow\mathbb{R}$ and \rev{the weights $\mathbf{r}$ of the vector \(R\)-symmetry \eqref{vectorRcharge} as }
$R_{j}\in\mathbb{R}$ for $j=1,\ldots,N=\mathrm{dim}_{\mathbb{C}}V$. We denote 
by $\mathcal{H}\subset\mathfrak{t}_{\mathbb{C}}$ the collection of hyperplanes 
\begin{eqnarray}\label{hperplanessigma}
  \mathcal{H}=\bigcup_{j=1}^{N}\bigcup_{n\in\mathbb{Z}_{\geq 
0}}\{\sigma\in \mathfrak{t}_{\mathbb{C}}~|~Q_{j}(\sigma)-\mathrm{i}\,\frac{R_{j}}{2}-\mathrm{i}\,n=0\}.
  \end{eqnarray}
The role of \(\mathcal{H}\) will become more transparent 
after we introduce the hemisphere partition function; it
is the locus we want to avoid in the integral (cf.~Definition \ref{definition:hemispherePF}).
\rev{The other piece of data that we need, termed \rev{$[L_{c}]$}, is a profile for the 
vector multiplet scalar (which, in the context of the physics of GLSM's corresponds to $\sigma$ in \eqref{hperplanessigma}).}
\begin{definition}
\label{definition:admissible-contour}
 Consider a gauge-invariant (i.e.~invariant under the $G$-action) 
middle-dimensional subvariety 
\rev{$L_{c}\subset \mathfrak{g}_{\mathbb{C}}\setminus \mathcal{H}$ }of the 
complexified Lie algebra of $  \mathsf{G} $ or equivalently its intersection 
\rev{$L_{c}\subset\mathfrak{t}_{\mathbb{C}}\setminus \mathcal{H}$} invariant under 
the action of the Weyl group $\mathcal{W}_{  \mathsf{G} }\subset   \mathsf{G} $. We define an 
\emph{admissible contour} as a contour
\rev{$L_{c}$} satisfying the following two conditions:\textbf{}
\begin{itemize}
\item[(a)] \rev{\(L_{c}\)} is a continuous deformation of the real contour 
$L_{\mathbb{R}}:=\{\Im\tau=0\mid \tau\in 
\mathfrak{t}_{\mathbb{C}}\}$;
\item[(b)] \rev{the imaginary part of the boundary effective twisted superpotential \cites{Herbst:2008jq,Hori:2013ika,hori2019notes}}
$$\widetilde{W}_{\text{eff},q}\colon\mathfrak{t}_{\mathbb{C}}\rightarrow \mathbb{C}$$ defined by \rev{
\begin{equation}
\label{twistedbdry}
\begin{aligned}
&\widetilde{W}_{\text{eff},q}(\sigma):= \left(\sum_{\alpha>0}\pm 
\mathsf{i}\,\pi\,\alpha(\sigma)\right)\\&-\left(\sum_j 
(Q_j(\sigma))\left(\log\left(\frac{\mathsf{i}\, Q_j(\sigma)}{\Lambda}
\right)-1\right)\right)-t(\sigma)+2\pi \mathsf{i}\, q(\sigma)        
\end{aligned}
\end{equation}
approaches $-\infty$ in all asymptotic directions of $L_{c}$ and for all the 
weights $q\in \mathfrak{t}^{\vee}$ of $\rho_{M}$. The sum $\sum_{\alpha >0}$ in \eqref{twistedbdry} denotes the 
the sum over positive roots $\alpha$ of $\mathsf{G}$ and the sign depend on the Weyl chamber 
in which $\Re(\sigma)$ lies. This term is absent on our case, since we will restrict to examples where $\mathsf{G}$ is abelian.}
\end{itemize}
\rev{Hence, we use the notation $(\mathcal{B},[L_{c}])$ where by $[L_{c}]$ we meant the class of continuous deformations of $L_{c}$ that are admissible, for $\mathrm{Re}(t)$ in a fixed chamber $C$ and for all weights $q$ of $\rho_{M}$ corresponding to $\mathcal{B}$.}
\end{definition}  

\begin{remark}
\rev{By general physics considerations it is expected that $[L_{c}]$ only depends on the chamber $C$ (see for example \cite{Hori:2013ika}) where $\zeta=\mathrm{Re}(t)$ belongs, rather than in $t$. The existence of an admissible contour \(L_{c}\) remains widely open
in higher dimensions, i.e., when \(\dim Z(  \mathsf{G} )\ge 2\), or when $\mathsf{G}$ is nonabelian. However for $\mathsf{G}$ abelian, $[L_{c}]$ has been explicitly constructed for models with a geometric phase in \cite{Aleshkin:2022eop}.
In Appendix \ref{app:contours}, we give an explicit construction of \(L_{c}\)
for the cases and chamber $C$ under consideration here.}
\end{remark}

\rev{\begin{definition}
\label{definition:fullMF-data}
 The objects $(\mathcal{B},[L_{c}])\in \mathrm{MF}_{  \mathsf{G} }(W)$ are composed from algebraic 
data $\mathcal{B}$ and the auxiliary data of (the class of deformations of) an admissible contour $[L_{c}]$. 
The morphisms between two objects are the same as in \cite{ballard2019variation} (specialized to the current case) and forgets the auxiliary data, namely
$$(\mathcal{B}_{1},[L_{c}]), (\mathcal{B}_{2},[L'_{c}])\in \mathrm{MF}_{  \mathsf{G} }(W),$$ 
is defined by the class of the maps
\begin{equation}
\Psi\in \mathrm{Hom}(M_{1},M_{2})
  \end{equation}
 in the cohomology defined by the differential $D$:
\begin{equation}
D\Psi:=\mathbf{T}_{2}\Psi-(-1)^{|\Psi|}\Psi \mathbf{T}_{1}
  \end{equation}
  where $|\Psi|\in\{ 0,1\}$ is the $\mathbb{Z}_{2}$-degree of $\Psi$. Note that this definition of morphisms is independent of $C$ (or the value of $t$). However, we implictly always consider objects $\mathcal{B},[L_{c}])$ with the auxliary data in a fixed chamber $C$.
\end{definition} } 
\rev{As mentioned above, the algebraic data $\mathcal{B}$ is defined in
\cite{ballard2019variation} 
in a more general context. However the category 
$\mathrm{MF}_{\mathsf{G}}(W)$ with objects $\mathcal{B}$ with an extra auxiliary data $[L_{c}]$, depending on $\zeta=\mathrm{Re}(t)$, is inspired by the dynamics of $B$-branes on GLSMs \cite{Hori:2013ika} and we are not aware 
of an analogous consideration in the mathematics literature. Eventually, this auxiliary data will play a central role on defining the $A$-periods of $B$-branes.}
\begin{remark}[Grade restriction rule]
The parameter $e^{t}$ corresponds to local coordinates on the stringy K\"{a}hler moduli space $\mathcal{M}_{\mbox{\footnotesize K\"{a}hler}}$. 
The space $\mathcal{M}_{\mbox{\footnotesize K\"{a}hler}}$ is locally isomorphic to $(\mathbb{C}^{\ast})^{\dim(\mathfrak{z})}$. Globally it corresponds to a partial compactification of $(\mathbb{C}^{\ast})^{\dim(\mathfrak{z})}\setminus \Delta$, where $\Delta$ is a complex codimension 1 closed subset. The space $\mathcal{M}_{\mbox{\footnotesize K\"{a}hler}}$ can be determined from the nonanomalous GLSM data using physical considerations \cite{Witten:1993yc} however we are not aware of a general and purely mathematical definition of $\mathcal{M}_{\mbox{\footnotesize K\"{a}hler}}$ that does not rely on indirect methods, such as mirror symmetry. When $  \mathsf{G} \cong U(1)^{s}$, then $\mathcal{M}_{\mbox{\footnotesize K\"{a}hler}}$ can be determined from toric geometry considerations plus a study of Coulomb and mixed Coulomb--Higgs branches of the GLSM
\cite{Morrison-Plesser}. This gives, in the abelian case, $\mathcal{M}_{\mbox{\footnotesize K\"{a}hler}}$ a chamber structure and the different chambers can be classified by cones in the secondary fan, as we will see in section \ref{subsec:toric-git}. Each of these chambers is termed a \emph{phase}, in the physics jargon and whenever $\zeta=\Re(t)$ belongs to the interior of a chamber (and $e^{t}\notin\Delta$) it is expected that there always exists a unique, up to homotopy, an 
admissible $L_{t}$ for any quadruple $\mathcal{B}$. 
We will show this explicitly by constructing \rev{$L_{c}$} for the examples (and the chamber) 
we are interested in in Appendix \ref{app:contours}. Hence, on a fixed phase, we can ignore this 
piece of data and work with an entirely algebraic category whose objects are 
the algebraic data $\mathcal{B}$. It is important to remark that this is not generically true if we consider paths in 
$t$-space. For a given $\mathcal{B}$ an admissible $L_{t}$ can stop being admissible as $t$ crosses certain regions (\emph{phase walls}, \rev{which approximately coincide with the chamber walls, determined by the secondary fan (in the case $\mathsf{G}$ abelian)}. This problem can be solved by taking the cone between $\mathcal{B}$ and some $\mathcal{B}_{nh}$ \rev{whose image under the functor \eqref{functorphase} is nullhomotopic}. This operation results in an equivalent algebraic data $\mathcal{B}'$, \rev{upon projection by \eqref{functorphase}, but now $L_{c}$ remains admissible as we cross the phase wall (which was not a necessary condition for admissibility as in Definition \ref{definition:admissible-contour})}. This phenomenon of `B-brane transport' along phase walls is known as the grade restriction rule and was originally studied and solved for $  \mathsf{G} $ abelian in \cite{Herbst:2008jq} and has been rigourosly formulated in the 
mathematics literature in 
\cites{segal2011equivalences,halpern2015derived,ballard2019variation} for general 
$G$. A physics perspective on the case of general group $  \mathsf{G} $, including also anomalous GLSMs, can be 
found in 
\cites{Clingempeel:2018iub,Hori:2013ika,hori2019notes,eager2017beijing}. In the present work we are only concerned with a specific phase for a family of abelian GLSMs we will define in section \ref{subsec:glsm-for-nc} therefore we will not be concerned with the grade restriction rule.
\end{remark}
We are ready to define our main function that will give rise to the $A$-periods. 

\begin{definition}
\label{definition:hemispherePF}
\rev{We define the central charge of an object $(\mathcal{B},[L_{c}])\in MF_{G}(W)$ to 
be the hemisphere partition function \cites{Hori:2013ika,Honda:2013uca,Sugishita:2013jca}
\begin{equation}
\label{ZD2}
Z_{D^{2}}(\mathcal{B}):=\int_{L_{c}\subset 
\mathfrak{t}_{\mC}} \mathrm{d}^{l_{  \mathsf{G} }}\sigma 
\prod_{\alpha>0}\alpha(\sigma)\sinh(\pi\alpha(\sigma))\prod_{j=1}
^{N}\Gamma\left(\mathrm{i}\,Q_
{j}(\sigma)+\frac{R_{j}}{2}\right)e^{\mathrm{i}\,t(\sigma)}f_{\mathcal{B}}(\sigma).
\end{equation}}
where
\begin{equation}
f_{\mathcal{B}}(\sigma):=\mathrm{tr}_{M}\left(R_{M}(e^{i
\pi})\rho_{M}(e^{2\pi\sigma
} )\right)
\end{equation}
the symbol $\prod_{\alpha>0}$ denotes the product over the positive roots of 
$G$, $l_{  \mathsf{G} }:=\mathrm{dim}(\mathfrak{t})$ \rev{and $L_{c}$ is any representative of the class $[L_{c}]$}. 
\end{definition}  
The function (\ref{ZD2}) was computed by direct supersymmetric localization 
methods in quantum field theory 
\cites{Hori:2013ika,Honda:2013uca,Sugishita:2013jca}.

\subsection{GLSMs and \texorpdfstring{$A$}{A}-periods}

In order to study the IR theory of a nonanomalous GLSM, characterized by the 
GLSM data $(  \mathsf{G} ,W,\rho_{m},t,R)$ we need to define the classical Higgs branch for 
a given value of $\zeta=\mathrm{Re}(t)$,
\begin{definition}
\label{definition:class-higgs}
We define the classical Higgs branch at $\zeta$, associated to the GLSM data 
$(  \mathsf{G} ,W,\rho_{m},t,R)$ as the pair $(Y_{\zeta},W_{\zeta})$
where
\begin{equation}\label{sympquotglsm}
Y_{\zeta}:=\mu^{-1}(\zeta)/  \mathsf{G} 
\end{equation}
where $\mu\colon V\rightarrow \mathfrak{g}^{\vee}$ denotes the moment map associated 
to $\rho_{m}$ and 
\begin{equation}
W_{\zeta}:=W|_{Y_{\zeta}}
\end{equation}
\end{definition}  
For $\zeta$ in a phase, there exists a projection functor $\pi_{\zeta}$ \cites{Herbst:2008jq,ballard2019variation}
\begin{equation}\label{functorphase}
\pi_{\zeta}\colon \mathrm{MF}_{  \mathsf{G} }(W)\rightarrow \mathcal{D}_{\zeta},
\end{equation}
$\mathcal{D}_{\zeta}$ is a triangulated 
category, known as the IR B-brane category, and its specific description depends on the particular model we are considering. 
The functor $\pi_{\zeta}$ and the category $\mathcal{D}_{\zeta}$  only depend on the chamber where $\zeta$ belongs.
For an object $\mathcal{E}\in\mathcal{D}_{\zeta}$ we can define the $A$-period 
$Z(\mathcal{E})$, where $Z\colon\mathcal{D}_{\zeta}\rightarrow\mathbb{C}$ is a 
function known as the central charge of $\mathcal{E}$. This function is related 
to Bridgeland's central charge \cite{MR2373143}, however its intrinsic definition, 
relevant for mirror symmetry, has been studied in the physics literature 
for general superconformal field theories \cites{Knapp:2020oba,hori2019notes}. In \cite{Aleshkin:2023hbu}, a mathematical proof of the relation between (\ref{ZD2}) and appropriately defined invariants over a space of quasimaps associated to the GLSM has been presented (such a space of quasimaps has been studied for instance in \cites{Ciocan-Fontanine:2018oaw,Fan:2015vca}).
\rev{For the case that the GLSM is nonanomalous (in the sense of Definition \ref{definition:glsm-data}) and $X_{\zeta}:=\mathrm{d}W^{-1}_{\zeta}(0)\cap Y_{\zeta}$ is a smooth Calabi--Yau manifold we have the identification
\begin{equation}
\mathcal{D}_{\zeta}=D^{b}Coh(X_{\zeta}),
\end{equation}
which was proven in \cites{isik2013equivalence,Shipman:2010umu} (and also see Section 4.4 of \cite{halpern2016autoequivalences}). In such case, a definition of \(A\)-periods has been proposed in the physics }
\cite{Hosono:2000eb} and mathematics \cite{Iritani:2009ab} literature. It is conjectured in  
\cite{Hori:2013ika} that
\begin{equation}
Z_{D^{2}}(\mathcal{B})=Z(\pi_{\zeta}(\mathcal{B})).
\end{equation}
We will then make use of this conjecture to define the \(A\)-periods in the following. For the cases when the stabilizer of the points 
$\mu^{-1}(\zeta)$ is a finite subgroup of $  \mathsf{G} $, we are in the physics situation 
known as a \emph{weakly coupled phase} and the IR B-brane category can be 
characterized as \cite{hori2019notes}
\begin{equation}
\mathcal{D}_{\zeta}=\mathrm{MF}(Y_{\zeta},W_{\zeta})
\end{equation}
where $\mathrm{MF}(Y_{\zeta},W_{\zeta})$ denotes the category of matrix 
factorizations \cites{buchweitz1987maximal,orlov2004triangulated,efimovcoherent,
ballard2012resolutions}. It is worth noting that this category is independent 
of the choice of $\zeta$ inside 
the chamber.

\subsection{GLSMs for noncommutative resolutions}\label{subsec:glsm-for-nc}

We will study a particular class of GLSMs, whose IR theory in a particular 
phase, i.e., the classical Higgs branch $(Y_{\zeta},W_{\zeta})$, for $\zeta$ in 
the interior of some cone, satisfies the following conditions 
\begin{itemize}
 \item The variety $Y_{\zeta}$ takes the form of a vector bundle 
$Y_{\zeta}=\mathcal{V}\rightarrow B$, where $B$ is compact K\"ahler and we 
allow the posibility of $B$ having orbifold singularities and $\mathcal{V}$ 
being a orbibundle. More precisely, one should consider $B$ an algebraic stack \cite{Pantev:2005zs}.
\item We require $W_{\zeta}\in 
\mathrm{H}^{0}(\mathcal{O}_{Y_{\zeta}})$ and 
\begin{equation}
\mathrm{d}W_{\zeta}^{-1}(0)=B
\end{equation}
as a set.
\item There exist a holomorphic vertical Killing vector field $\xi$ in 
$Y_{\zeta}$ implementing a $\mathbb{C}^{\ast}$ action on $Y_{\zeta}$ and satisfying
\begin{equation}
\mathcal{L}_{\xi}W_{\zeta}=W_{\zeta}
\end{equation}
\end{itemize}
these conditions characterize what is called a \emph{good hybrid} 
\cites{Aspinwall:2009qy,Bertolini:2013xga} (originally studied in 
\cite{Witten:1993yc}).
We will be particularly interested in the case when the function $W_{\zeta}$ is 
quadratic in the fiber coordinates of $\mathcal{V}$. Let us outline the construction suitable for our singular double covers. A toric description of the variety $X$ can be given in terms \rev{of} homogeneous coordinates $\phi_{I}$ i.e.~the Cox ring of $X$. These coordinates are weighted by the action of the gauge group. Denote their $  \mathsf{G} $-weights by $\theta_{I}$. 
\rev{Instead of considering $X$, we will consider the GLSM describing a gerbe over $X$. This amounts to considering the toric variety $\widetilde{X}$ described by homogeneous coordinates $\tilde{\phi}_{I}$ whose weights are $2\theta_{I}$. Then there is a trivially acting subgroup of $\mathsf{G}$. Consider a collection of line bundles $E_{a}$ over $X$, then any section $s_{a}(\phi)\in \mathrm{H}^0(X,E_{a})$ can be thought of as a section of
a bundle over $\widetilde{X}$ by considering $f_{a}(\tilde{\phi}):=s_{a}(\tilde{\phi})$ with the weights 
being doubled. The $\mathsf{G}$-weights of $f_{a}(\tilde{\phi})$ will be a multiple of $2$ and therefore it makes sense to consider the line bundles $\widetilde{E}^{-1/2}_{a}\rightarrow\widetilde{X}$, where $\widetilde{E}_{a}$ denotes the lift of $E_{a}$ to $\widetilde{X}$. Then $\widetilde{E}_{a}$ has a square root: $(\widetilde{E}_{a}^{-1/2})^{\otimes 2}=\widetilde{E}_{a}^{-1}$. Moreover $z_{a}^{2}f_{a}(\tilde{\phi})\in \mathrm{H}^{0}(\mathcal{O}_{\widetilde{X}})$, where $z_{a}$ is the coordinate on the fiber of $\widetilde{E}^{-1/2}_{a}$.
}

\rev{
The hybrid model we want to consider has the target space  
\(Y_{\zeta}\cong\mathcal{V}\)
where
\begin{equation}
\mathcal{V}= \operatorname{Tot}(\oplus_{a}\widetilde{E}^{-1/2}_{a}\rightarrow \widetilde{X})
\end{equation}
is the total space of the vector bundle 
over the gerbe \(\widetilde{X}\) and $E_{a}$'s are line bundles 
over $X$.}
Therefore, locally, if we denote the fiber coordinates as $z_{a}$, 
$a=1,\ldots, \mathrm{rk}(\mathcal{V})$, and the base coordinates in the gerbe as $\tilde{\phi}$ then 
$W_{\zeta}$ can be 
written as
\begin{equation}
W_{\zeta}=\sum_{a=1}^{\mathrm{rk}(\mathcal{V})}z_{a}^{2}f_{a}(\tilde{\phi})
\end{equation}
where \rev{$f_{a}(\tilde{\phi})\in \mathrm{H}^{0}(\widetilde{X},\widetilde{E}_{a})$ as explained above}.
The resulting quantum field theory is interpreted as a 
noncommutative (NC) resolution of the (possibly singular) variety given 
by a double cover of $X$ branched along $\prod_{a}f_{a}(\phi)$ 
\cite{Caldararu:2010ljp}. Equivalently we can write
\rev{
\begin{equation}
Y_{\zeta}\cong \operatorname{Tot}(\oplus_{a}\widetilde{E}_{a}^{-\frac{1}{2}}\rightarrow\widetilde{X}).
\end{equation}
}
This description of $Y_{\zeta}$ as a bundle over a gerbe appeared originally in \cite{Pantev:2005zs} in the context of GLSMs and in \cite{Caldararu:2010ljp} it was identified with a NC resolution for double covers of $\mathbf{P}^3$. More recently, in \cites{Katz:2022lyl,Katz:2023zan} the authors have presented a detailed study of the enumerative geometry of CY (singular) double covers over
$\mathbf{P}^3$
\rev{as well as} their noncommutative resolutions. It is important to remark that even though we can consider $\tilde{X}$ as an orbifold (with trivial action), it is crucial to consider $\tilde{X}$ as an algebraic stack, in order to identify the correct category of $B$-branes \cites{Addington:2012zv,Guo:2021aqj}. We will give more details on the GLSM construction, its $B$-branes and their $A$-periods in section \ref{subsubsec:a-period-glsm}.

\subsection{Abelian GLSMs}
\label{subsec:abelian}

Since we will be working on 
abelian GLSMs, in this subsection, 
we will apply the theory in \S\ref{sec:GLSMmir} to 
abelian (and nonanomalous) GLSMs. For our purpose, we also assume that \(  \mathsf{G} \) is connected
in what follows,
i.e., \(  \mathsf{G} =\mathrm{U}(1)^{s}\).
In this case, we have
\begin{itemize}
\item \( \mathsf{G}=T\) and \(W_{ \mathsf{G}}=1\);
\item \(\pi_{1}(G)\cong\mathbb{Z}^{s}\);
\item\(\mathfrak{t}\cong\mathbb{R}^{s}\),
\(\mathfrak{t}_{\mathbb{C}}=\mathfrak{t}\otimes_{\mathbb{R}}\mathbb{C}\cong\mathbb{C}^{s}\),
and the weight lattice \(\mathrm{P}\cong\mathbb{Z}^{s}\).
The set of the Fayet--Iliopolous (FI)-theta parameter is given by
\begin{equation*}
\mathfrak{t}^{\vee}_{\mathbb{C}}\slash 2\pi\mathrm{i} \mathrm{P}\cong \mathbb{R}^{s}\oplus (\mathrm{i}\mathbb{R}\slash 2\pi
\mathrm{i}\mathbb{Z})^{s}.
\end{equation*}
\end{itemize}
The hemisphere partition function \eqref{ZD2} is simplified to
\begin{equation}
\label{eq:partition-fcn-abelian}
Z_{D^{2}}(\mathcal{B}) = \int_{L_{t}} \mathrm{d}^{s}\sigma 
\prod_{j=1}^{N} \Gamma\left(\mathrm{i}Q_{j}(\sigma)+\frac{R_{j}}{2}\right) 
e^{\mathrm{i}t(\sigma)} f_{\mathcal{B}}(\sigma).
\end{equation}
Here \(\sigma=(\sigma_{1},\ldots,\sigma_{s})\) is
the coordinate on \(\mathfrak{t}_{\mathbb{C}}\) and 
\(t=(t_{1},\ldots,t_{s})\in \mathfrak{t}_{\mathbb{C}}^{\vee}\slash 2\pi\mathrm{i} \mathrm{P}\).
We will also use the notation \(t=\zeta-\mathrm{i}\theta\)
with \(\zeta=(\zeta_{1},\ldots,\zeta_{s})\in\mathbb{R}^{s}\) and 
\(\theta=(\theta_{1},\ldots,\theta_{s})\in (\mathrm{i}\mathbb{R}\slash 2\pi
\mathrm{i}\mathbb{Z})^{s}\). 

From the construction, the central charge \(Z_{D^{2}}(\mathcal{B})\)
is apparently a multi-valued function in \(t\). However,
it would be more convenient to regard \(\theta\) as
the coordinate on the universal cover \(\mathbb{R}^{s}\)
and \(Z_{D^{2}}(\mathcal{B})\) as a function on \(\mathfrak{t}_{\mathbb{C}}\)
rather than a multi-valued function on the 
quotient \(\mathfrak{t}_{\mathbb{C}}^{\vee}\slash 2\pi\mathrm{i} \mathrm{P}
\cong (\mathrm{i}\mathbb{R}\slash 2\pi
\mathrm{i}\mathbb{Z})^{s}\). Accordingly, we will write
\begin{equation*}
t(\sigma) = \langle t,\sigma\rangle = \sum_{k=1}^{s} t_{k}\sigma_{k}
\end{equation*}
the usual canonical dual pairing between \(\mathfrak{t}_{\mathbb{C}}\) and 
\(\mathfrak{t}_{\mathbb{C}}^{\vee}\).
Consequently, 
\begin{equation*}
e^{\mathrm{i}t(\sigma)} = 
e^{\mathrm{i}\langle t,\sigma\rangle} = 
\prod_{k=1}^{s} (e^{-t_{k}})^{-\mathrm{i}\sigma_{k}}=
\prod_{k=1}^{s} q_{k}^{-\mathrm{i}\sigma_{k}}
\end{equation*}
and each factor is a multi-valued function on \(\mathbb{C}^{\ast}\)
(with coordinate \(q_{k}\)).

\subsection{Toric GIT and the secondary fan}
\label{subsec:toric-git}
For the sake of completeness, we recall some basics of toric GIT. Details can be
found in \cite{Cox-toric-varieties}*{Chapter 14}.
We recall the definition of the secondary fan and we relate it with the concept of phases of a GLSM.

Consider an
algebraic \emph{subgroup} \(G\subset (\mathbb{C}^{\ast})^{r}\). It is known that the induced map on the character groups
\begin{equation}
\label{eq:char-morphism}
\begin{tikzcd}
&\mathbb{Z}^{r}\ar[r] &\widehat{G}
\end{tikzcd}
\end{equation}
is surjective. Note that we do not assume that \(G\) is an algebraic torus; 
\(G\) might have non-trivial torsion elements. Denote by \(M\) its kernel. We obtain a 
short exact sequence
\begin{equation}
\label{eq:char-exact-seq-z}
\begin{tikzcd}
& 0\ar[r] &M\ar[r] &\mathbb{Z}^{r}\ar[r] &\widehat{G}\ar[r]&0.
\end{tikzcd}
\end{equation}
Let \(G\) act on \(\mathbb{C}^{r}\) via the inclusion \(G\subset (\mathbb{C}^{\ast})^{r}\).
Each \(\tau\in \widehat{G}\) determines a GIT quotient \(\left[\mathbb{C}^{r}\sslash_{\tau} G\right]\).
Applying \(\mathrm{Hom}_{\mathbb{Z}}(-,\mathbb{Z})\) to \eqref{eq:char-exact-seq-z}, we get
\begin{equation}
\label{eq:char-exact-seq-z-dual}
\begin{tikzcd}[column sep=1em]
& 0\ar[r] &\mathrm{Hom}_{\mathbb{Z}}(\widehat{G},\mathbb{Z})\ar[r] 
&\mathbb{Z}^{r}\ar[r] &N:=\mathrm{Hom}_{\mathbb{Z}}(M,\mathbb{Z})\ar[r]
&\mathrm{Ext}^{1}_{\mathbb{Z}}(\widehat{G},\mathbb{Z})\ar[r]&0.
\end{tikzcd}
\end{equation}
The morphism \(\mathbb{Z}^{r}\to N\) has a finite cokernel.
For \(\tau\in\widehat{G}\), it is known that 
the underlying space of the GIT quotient \(\left[\mathbb{C}^{r}\sslash_{\tau} G\right]\) is a toric variety.
To describe the underlying space, let \(\nu_{i}\) be the image of \(e_{i}=(\delta_{1i},\ldots,\delta_{ri})\in\mathbb{Z}^{r}\)
under the morphism \(\mathbb{Z}^{r}\to N\) in \eqref{eq:char-exact-seq-z-dual}.

For each \(\tau\in\widehat{G}\), we choose a lifting \(\mathbf{a}:=(a_{1},\ldots,a_{r})\in\mathbb{Z}^{r}\)
from \eqref{eq:char-morphism}; this is possible since \eqref{eq:char-morphism} is surjective.
Define the polyhedron
\begin{equation}
P_{\mathbf{a}}:=\{m\in M_{\mathbb{R}}~|~\langle m,\nu_{i}\rangle \ge -a_{i},~i=1,\ldots,r\}.
\end{equation}
It is known that the toric variety associated to \(P_{\mathbf{a}}\) is the
underlying toric variety of the GIT quotient \(\left[\mathbb{C}^{r}\sslash_{\tau} G\right]\)
(cf.~\cite{Cox-toric-varieties}*{Theorem 14.2.13}).

\begin{instance}
\label{ex:git}
Let \(G =\{(t,s)~|~s^{2}=t^{2}\}\subset (\mathbb{C}^{\ast})^{2}\) acting on \(\mathbb{C}^{2}\).
Then \(\mathbb{C}^{\ast}\times \mu_{2}\cong G\)
(via \((t,\zeta_{2})\mapsto (t,\zeta_{2}t)\)).
Here \(\mu_{2}\) is the cyclic group 
of order two, written additively.
Under this identification, the induced map between their character group is
\begin{equation}
\mathbb{Z}^{2} \to \widehat{G}=\mathbb{Z}\times\mu_{2},~~(a,b)\mapsto (a+b, b\mod{2})
\end{equation}
whose kernel is given by \(M=\{(a,b)~|~a+b=0,~b\in2\mathbb{Z}\}\cong\mathbb{Z}\). 
Applying \(\mathrm{Hom}_{\mathbb{Z}}(-,\mathbb{Z})\) to
\begin{equation}
\label{eq:char-ext}
\begin{tikzcd}
& 0 \ar[r] &M\ar[r]& \mathbb{Z}^{2} \ar[r] &\widehat{G}\ar[r] &0,
\end{tikzcd}
\end{equation}
we obtain an exact sequence
\begin{equation}
0\to\mathbb{Z}=\mathrm{Hom}_{\mathbb{Z}}(\widehat{G},\mathbb{Z})
\xrightarrow{a\mapsto (a,a)}\mathbb{Z}^{2} \xrightarrow{(x,y)\mapsto 2x-2y}  N\cong\mathbb{Z} \to \mu_{2}\to 0.
\end{equation}
Let \(\nu_{i}\in N\) denote the image of \(e_{i}\in \mathbb{Z}^{2}\).
Consider the following various situations.
\begin{itemize}
\item[(a)] \(\tau=(1,1)\in\widehat{G}=\mathbb{Z}\times\mu_{2}\). 
We can choose a lifting of \((1,1)\) from \eqref{eq:char-ext}. 
For example, \(\mathbf{a}=(0,1)\in\mathbb{Z}^{2}\) will do. Then consider the polytope in \(M_{\mathbb{R}}\cong\mathbb{R}\)
\begin{equation}
P_{\mathbf{a}}=\{m\in M_{\mathbb{R}}~|~\langle m,\nu_{1}\rangle\ge 0,~\langle m,\nu_{2}\rangle\ge -1\}=\{m\in M_{\mathbb{R}}~|~0\le m\le 1/2\}
\end{equation}
whose normal fan \(\Sigma\) is the fan defining \(\mathbf{P}^{1}\). The triple \((N,\Sigma,\mathbb{Z}^{2}\to N)\)
is the \emph{stacky fan} in the sense of Borisov--Chen--Smith
\cite{2005-Borisov-Chen-Smith-the-orbifold-chow-ring-of-toric-deligne-mumford-stacks}.
The quotient stack we obtain is \(\left[\mathbb{C}^{2}\sslash_{\tau}G\right]\cong \left[\mathbf{P}^{1}\slash \mu_{2}\right]\)
whose coarse moduli space is equal to \(\mathbf{P}^{1}\).
\item[(b)] \(\tau=(1,0)\in\mathbb{Z}\times\mu_{2}\). In this case, we can use \(\mathbf{a}=(1,0)\in\mathbb{Z}^{2}\) as our lifting.
It turns out that we obtain the same toric stack in (a).
\item[(c)] \(\tau = (0,1)\in\mathbb{Z}\times \mu_{2}\). Let \(\mathbf{a}=(-1,1)\) be a lifting. In this case, we have
\begin{equation}
P_{\mathbf{a}}=\{m\in M_{\mathbb{R}}~|~\langle m,\nu_{1}\rangle\ge -1,~\langle m,\nu_{2}\rangle \ge 1\}=\{1/2\}
\end{equation}
and the underlying space of the GIT quotient \(\left[\mathbb{C}^{2}\sslash_{\tau} G\right]\) is indeed a point.
\item[(d)] \(\tau=(-1,0)\in\mathbb{Z}\times \mu_{2}\). We can pick a lifting \(\mathbf{a}=(-1,1)\in\mathbb{Z}^{2}\)
and the polyhedron in the present situation is 
\begin{equation}
P_{\mathbf{a}}=\{m\in M_{\mathbb{R}}~|~\langle m,\nu_{1}\rangle\ge 1,~\langle m,\nu_{2}\rangle \ge 0\}=\emptyset
\end{equation}
so \(\left[\mathbb{C}^{2}\sslash_{\tau} G\right]=\emptyset\) in this case, i.e., 
there is no semistable point under the stability condition \(\tau=(-1,0)\).
\end{itemize}
We can see the wall-crossing phenomenon clearly in this example.
Notice that in the present case, these are three inequivalent GIT quotients and these are all 
possibilities.
(Recall that (a) and (b) are equivalent.) This is because there are only three
cones (two chambers) in the secondary fan.
\end{instance}

Note that given \(G\subset (\mathbb{C}^{\ast})^{r}\) there are only finitely many distinct GIT quotients 
\(\left[\mathbb{C}^{r}\sslash_{\tau} G\right]\) up to isomorphism. Denote by
\(\chi_{i}\) the image of \(e_{i}=(\delta_{1i},\ldots,\delta_{ri})\in\mathbb{Z}^{r}\)
under the morphism \(\mathbb{Z}^{r}\to \widehat{G} \to \widehat{G}_{\mathbb{R}}\)
and let
\begin{equation}
C_{\boldsymbol\chi}:=\operatorname{Cone}(\chi_{1},\ldots,\chi_{r})
\end{equation}
be the cone generated by \(\chi_{i}\) in the vector space \(\widehat{G}_{\mathbb{R}}\).
It is known that there is a fan, called
the \emph{secondary fan}, whose support is \(C_{\boldsymbol\chi}\) and
satisfying the following property:
\begin{itemize}
\item the collection of the relative interior of cones in the secondary fan gives rise to a decomposition
of \(C_{\boldsymbol\chi}\) such that the GIT quotient \(\left[\mathbb{C}^{r}\sslash_{\tau} G\right]\)
is constant on each subset.
\end{itemize}
We can think of the secondary fan as the space parameterizing all the possibly non-empty GIT quotients.

\begin{instance}[Example \ref{ex:git} continued]
We have \(\widehat{G}_{\mathbb{R}}=\mathbb{R}\) in this case. 
The secondary fan (under the identification \(\mathbb{C}^{\ast}\times\mu_{2}\cong G\)) consists of two cones
\begin{equation}
\mathbb{R}_{\ge 0}~\mbox{and}~\{0\}
\end{equation}
and has the ``chamber decomposition'' \(\mathbb{R}_{\ge 0} = \{0\}\cup \mathbb{R}_{>0}\).
\end{instance}



\section{
\texorpdfstring{\(A\)}{A}-periods of GLSM and 
\texorpdfstring{\(B\)}{B}-periods of \texorpdfstring{\(Y^{\vee}\)}{Yvee}}
\label{sec:Periods} 


\subsection{GKZ systems for periods of Calabi--Yau double covers}
\label{subsec:gkz}
Given an integral matrix \(A=(a_{ij})\in\mathrm{Mat}_{d\times m}(\mathbb{Z})\) and
a parameter \(\beta=(\beta_{i})\in\mathbb{C}^{d}\), 
the GKZ system \(\mathcal{M}_{A}^{\beta}\)
is a set of partial differential equations on \(\mathbb{C}^{m}\) consisting
of the following two types of operators
\begin{itemize}
\item the box operators \(\Box_{\ell}:=\partial^{\ell^{+}}-\partial^{\ell^{-}}\)
where \(\ell^{\pm}\in \mathbb{Z}^{m}_{\ge 0}\) such that \(A\ell^{+}=A\ell^{-}\);
\item the Euler operator \(E_{i}-\beta_{i}:=\sum_{j=1}^{m} a_{ij}x_{j}\partial_{j}-
\beta_{i}\) for \(i=1,\ldots,d\).
\end{itemize}
In the bullets, \(x_{1},\ldots,x_{m}\) are coordinates on \(\mathbb{C}^{m}\)
corresponding to the columns of \(A\),
\(\partial_{j}\equiv\partial\slash\partial x_{j}\) is the partial
derivative, and 
the multi-index convention is used.

Suppose we are given the data in \S\ref{subsection:b-b-construction} and
let us retain the notation there.
It is known that the period integrals
of \(\mathcal{Y}^{\vee}\to V\) satisfy a certain type of 
GKZ systems with a fractional exponent. 
\begin{itemize}
\item[(6)] Let \(\Sigma\) be the fan defining \(X\) and
\(J_{1}\sqcup\cdots\sqcup J_{r}=\Sigma(1)\) be the corresponding 
nef-partition on \(X\). Put \(J_{k}=\{\rho_{k,1},\ldots,\rho_{k,m_{k}}\}\).
Also we put additionally \(\rho_{k,0}=\mathbf{0}\in N\) for each \(k=1,\ldots,r\).
From the duality construction, we have
\(\nabla_{k}\cap N = J_{k}\cup\{\rho_{k,0}\}\).
We also use the same notation \(\rho_{i,j}\)
to denote the primitive generator of the
corresponding \(1\)-cone in \(\Sigma\).

\item[(7)] 
Denote by \(\{e_{1},\ldots,e_{r}\}\) the standard basis of \(\mathbb{R}^{r}\).
For \(1\le i\le r\) and \(0\le j\le m_{i}\), we put \(\mu_{i,j}:=(e_{i},\rho_{i,j})\in\mathbb{Z}^{r+n}\).
Regarding \(\mu_{i,j}\) as column vectors, we define for each \(i=1,\ldots,r\) a matrix
\begin{equation*}
A_{i} = \begin{bmatrix}
\vline height 1ex & & \vline height 1ex \\
\mu_{i,0} & \cdots & \mu_{i,m_{i}}\\
\vline height 1ex & & \vline height 1ex 
\end{bmatrix}\in \mathrm{Mat}_{(r+n)\times (m_{i}+1)}(\mathbb{Z})
\end{equation*}
and put \(p=m_{1}+\cdots+m_{r}\) and
\begin{equation*}
A = \begin{bmatrix}
A_{1} & \cdots & A_{r}
\end{bmatrix}\in\mathrm{Mat}_{(r+n)\times (r+p)}(\mathbb{Z}).
\end{equation*}
We also set 
\begin{equation*}
\beta = \begin{bmatrix}
-1/2\\
\vdots\\
-1/2\\
\vline height 1ex\\
\mathbf{0}\\
\vline height 1ex
\end{bmatrix}\in \mathbb{Q}^{r+n},
\end{equation*}
where the first $r$ entries are $-\frac{1}{2}$. It will be easier to label the columns of \(A\) by
\((i,j)\) where \(1\le i\le r\) and
\(0\le j\le m_{i}\); the \((i,j)\)\textsuperscript{th}
column of \(A\) is precisely the vector \(\mu_{i,j}\).
\item[(8)] Given the matrix \(A\) and \(\beta\) as in 
\S\ref{subsec:gkz} (7), we denote by \(\mathcal{M}_{A}^{\beta}\)
the associated GKZ system.
The variables in the GKZ systems \(\mathcal{M}_{A}^{\beta}\) are
called \(x_{i,j}\) where \(1\le i\le r\) and \(0\le j\le m_{i}\);
it is the variable corresponding to the \((i,j)\)\textsuperscript{th}
column of \(A\).
\end{itemize}

The next proposition follows from a direct calculation. For the precise formulation
for period integrals and a proof,
see \cite{2022-Lee-Lian-Yau-on-calabi-yau-fractional-complete-intersections}.
\begin{proposition}
The period integrals of \(\mathcal{Y}^{\vee}\to V\)
satisfy the GKZ system \(\mathcal{M}_{A}^{\beta}\).
\end{proposition}

\begin{instance}[Example \ref{ex:p3} continued]
Let us explicitly write down the GKZ system for the periods of
\(\mathcal{Y}^{\vee}\to V\). From the construction, 
\(\nabla_{i}\) is the divisor polytope of \(F_{i}\) and 
\begin{equation}
\dim\mathrm{H}^{0}(X^{\vee},F_{i}) = 2,~i=1,2,3,4.
\end{equation}
Moreover, 
\begin{align*}
\nabla_{1} \cap N &= \{\mathbf{0},(1,0,0)\},\\
\nabla_{2} \cap N &= \{\mathbf{0},(0,1,0)\},\\
\nabla_{3} \cap N &= \{\mathbf{0},(0,0,1)\},\\
\nabla_{4} \cap N &= \{\mathbf{0},(-1,-1,-1)\}.
\end{align*}
The GKZ system \(\mathcal{M}_{A}^{\beta}\) describing the double cover family \(\mathcal{Y}^{\vee}\to V\) is given by 
\begin{equation}
A = \begin{bmatrix}
1 & 1 & 0 & 0 & 0 & 0 & 0 & 0\\
0 & 0 & 1 & 1 & 0 & 0 & 0 & 0\\
0 & 0 & 0 & 0 & 1 & 1 & 0 & 0\\
0 & 0 & 0 & 0 & 0 & 0 & 1 & 1\\
0 & 1 & 0 & 0 & 0 & 0 & 0 &-1\\
0 & 0 & 0 & 1 & 0 & 0 & 0 &-1\\
0 & 0 & 0 & 0 & 0 & 1 & 0 &-1\\
\end{bmatrix}~\mbox{and}~
\beta = 
\begin{bmatrix}
-1/2\\
-1/2\\
-1/2\\
-1/2\\
0\\
0\\
0
\end{bmatrix}.
\end{equation}

We also point out that in the present case, the GKZ system
is equivalent to Picard--Fuchs equations; in other words,
all the solutions to \(\mathcal{M}_{A}^{\beta}\)
are period integrals \cite{Lee-period-integrals}.
\end{instance}

\subsection{The \texorpdfstring{\(A\)}{A}-periods of the non-commutative resolutions}

\label{subsubsec:a-period-glsm}
Let us review some basics of toric varieties.
A standard reference is \cite{Cox-toric-varieties}.
We will follow the notation introduced in (1)\ndash(5) in the
beginning of \S\ref{sec:SingCY} and in (6)\ndash(8) in \S\ref{subsec:gkz}.

Let \(X\) be a smooth projective toric variety. We have
the short exact sequence 
\begin{equation}
\label{eq:basic-ses}
\begin{tikzcd}
&0\ar[r] &M\ar[r] &\mathbb{Z}^{p}\ar[r] &\operatorname{Cl}(X)\ar[r] &0.
\end{tikzcd}
\end{equation}
Here \(\operatorname{Cl}(X)\) is the Weil divisor class group of \(X\)
which equals the Cartier divisor class group owing to our hypothesis on \(X\).
Taking dual \(\operatorname{Hom}_{\mathbb{Z}}(-,\mathbb{Z})\) yields
\begin{equation}
\label{eq:basic-ses-dual}
\begin{tikzcd}
&0\ar[r] &\operatorname{Hom}_{\mathbb{Z}}(\operatorname{Cl}(X),\mathbb{Z})\ar[r] 
&\mathbb{Z}^{p}\ar[r,"B"] &N\ar[r] &0,
\end{tikzcd}
\end{equation}
where the matrix \(B\) is given by 
\begin{equation}
\label{eq:b-definition}
B = \begin{bmatrix}
\vline height 1ex & & \vline height 1ex \\
\rho_{1,1} & \cdots & \rho_{r,m_{r}}\\
\vline height 1ex & & \vline height 1ex 
\end{bmatrix}\in\mathrm{Mat}_{n\times p}(\mathbb{Z}).
\end{equation}
Applying \(\operatorname{Hom}_{\mathbb{Z}}(-,\mathbb{C}^{\ast})\)
to \eqref{eq:basic-ses}, 
we obtain the short exact sequence 
\begin{equation}
\label{eq:basic-ses-group}
\begin{tikzcd}
&1 \ar[r] &\operatorname{Hom}_{\mathbb{Z}}(\operatorname{Cl}(X),\mathbb{C}^{\ast})\ar[d,equal]
\ar[r] &(\mathbb{C}^{\ast})^{p} \ar[r] &T_{N}\ar[r] &1.\\
& &G & & &
\end{tikzcd}
\end{equation}


In order to describe the \(A\)-periods, we 
now introduce the following notation.

\begin{itemize}
\item[(9)] We fix once and for all an integral basis of \(\operatorname{Cl}(X)\)
consisting of ample divisors
and hence an isomorphism \(\operatorname{Cl}(X)\cong\mathbb{Z}^{s}\). Under this basis,
the third morphism in \eqref{eq:basic-ses} is represented by 
an integral matrix
\begin{equation*}
\begin{bmatrix}
\rule[.5ex]{2.5ex}{0.5pt} & \theta^{1} & \rule[.5ex]{2.5ex}{0.5pt}\\
 & \vdots& \\
\rule[.5ex]{2.5ex}{0.5pt} & \theta^{s} & \rule[.5ex]{2.5ex}{0.5pt}
\end{bmatrix}:=
\begin{bmatrix}
\theta^{1}_{1,1} & \cdots & \theta^{1}_{r,m_{r}}\\
\vdots & \ddots & \vdots\\
\theta^{s}_{1,1} & \cdots & \theta^{s}_{r,m_{r}}
\end{bmatrix}\in\mathrm{Mat}_{s\times p}(\mathbb{Z}).
\end{equation*}
Consequently, the character matrix of the 
second morphism in \eqref{eq:basic-ses-group}
is given by its transpose; in other words,
\(G=(\mathbb{C}^{\ast})^{s}\to (\mathbb{C}^{\ast})^{p}\) is given by
\begin{equation*}
(g_{1},\ldots,g_{s})\mapsto 
\left(\prod_{i=1}^{s}g_{i}^{\theta^{i}_{1,1}},\ldots,
\prod_{i=1}^{s}g_{i}^{\theta^{i}_{r,m_{r}}}\right)
\end{equation*}
and we have for all \(1\le k\le s\)
\begin{equation*}
\sum_{i=1}^{r}\sum_{j=1}^{m_{i}} \theta^{k}_{i,j}\rho_{i,j}=0.
\end{equation*}
We shall remind the reader that the \(\theta\) here is not
the \(\theta\) in the FI-parameter.
\item[(10)] Let \(\phi_{i,j}\), \(1\le i\le r\) and \(1\le j\le m_{i}\),
be the homogeneous coordinates of \(X\)
associated with the divisor \(\rho_{i,j}\). Under the basis chosen in (9),
\(\phi_{i,j}\) has weight \((\theta^{1}_{i,j},\ldots,\theta^{s}_{i,j})\).
\end{itemize}

\subsubsection{The abelian GLSM associated to CY double covers}

We propose a `curved' Landau--Ginzburg superpotential \(W\) to 
define the GLSM for
singular Calabi--Yau double covers. 

\begin{definition}[Curved Landau--Ginzburg superpotentials]
\label{def:curved-superpot}
With the data given as above, we
define a \emph{curved Landau--Ginzburg superpotential}
for the Calabi--Yau double cover over \(X\) to be
the function
\begin{equation}
\label{eq:superpot}
W = \sum_{i,j} z_{i,j}^{2}\phi_{i,j} + \sum_{k=1}^{r} z_{k,0}^{2} f_{k}(\phi),
\end{equation}
where \(f_{k}(\phi)\in\mathrm{H}^{0}(X,E_{k})\) is a section,
regarded as a homogeneous polynomial in \(\phi\)'s.
\end{definition}

\begin{remark}
The potential \eqref{eq:superpot} represents a family of LG models, with quadratic 
superpotential given by the coordinates $z_{i,j}$ and $z_{k,0}$ and fibered over $X$. 
The particular form of \eqref{eq:superpot} is chosen such that the
following property holds: the mass matrix of the $z$ 
variables has rank strictly less than $p+r$ if and only if 
$$
\left(\prod_{i,j}\phi_{i,j}\right)\prod_{k=1}^{r}f_{k}(\phi)=0
$$
which coincides with the branching locus of $Y$ (cf.~Remark \ref{rmk:branching-locus}). 
This idea was used originally in \cite{Caldararu:2010ljp} for $X=\mathbf{P}^3$.
\end{remark}
As we shall see, the notation \(z_{k,0}\) will be 
much more convenient for us and will give a more concise formula at the end.
Recall that \(E_{k}=\sum_{j=1}^{m_{k}} D_{k,j}\) and sections of
\(E_{k}\) correspond to a \((G,\theta_{k,0})\)-equivariant algebraic
function \(f_{k}(\phi)\)
on \(\mathbb{C}^{p}\) with 
\begin{equation}
\label{eq:theta-k-definition}
\theta_{k,0} \equiv (\theta^{1}_{k,0},\ldots,\theta^{s}_{k,0})
:= \sum_{j=1}^{m_{k}} (\theta^{1}_{k,j},\ldots,\theta^{s}_{k,j})
\end{equation}
We choose the \(G\)-weight of \(\phi_{i,j}\) to be $(\theta^{1}_{i,j},\ldots,\theta_{i,j}^{s})$ rescaled by 2; namely
\(\phi_{i,j}\) is equipped with \(G\)-weight
\begin{equation}
\label{eq:w-weight}
2(\theta^{1}_{i,j},\ldots,\theta_{i,j}^{s}).
\end{equation}
Then, to make \eqref{eq:superpot} \(G\)-invariant, the \(G\)-weight of \(z_{i,j}\) and \(z_{k,0}\)
are chosen as 
\begin{equation}
\label{eq:weight-sections}
\begin{cases}
\textstyle-(\theta^{1}_{i,j},\ldots,\theta^{s}_{i,j}),&~\mbox{for}~z_{i,j},\\
\textstyle-(\theta^{1}_{k,0},\ldots,\theta^{s}_{k,0}),
&~\mbox{for}~z_{k,0}.
\end{cases}
\end{equation}
\begin{lemma}
Under the weight assignments \eqref{eq:w-weight}
and \eqref{eq:weight-sections},  
\(W\) becomes \(G\)-invariant.
\end{lemma}

By the discussion in section \ref{subsec:glsm-for-nc}, the category of B-branes of this GLSM in the nc resolution phase is given by
\begin{eqnarray}
\mathrm{MF}(Y_{\zeta},W_{\zeta})
\end{eqnarray}
where 
\rev{
\begin{eqnarray}
\label{orbibundle}
Y_{\zeta}\cong 
\mathcal{V}\bigoplus_{k=1}^{r}\mathcal{L}_{(k,0)}\rightarrow 
[X/\mu^{s}_{2}]\qquad 
\mathcal{V}:=\bigoplus_{i=1}^{r}\bigoplus_{j=1}^{m_{i}}\mathcal{L}_{(i,j)}
\end{eqnarray}
}and the line bundles are orbibundles given by
\begin{equation}
\mathcal{L}_{(i,j)}=\mathcal{O}_{X}(-\theta^{1}_{i,j}/2,\ldots,-\theta^{s}_{i,j}
/2 )~\mbox{and}~\mathcal{L}_{(k,0)}=\mathcal{O}_{X}(-\theta^{1}_{k,0}/2,\ldots,-\theta^{s}_{k,0}
/2 )
\end{equation}
\rev{
Our curved superpotential \(W_{\zeta}\) 
is equipped with a \(\mu_{2}^{s}\)-symmetry
given by fiberwise action;
hence, \(W_{\zeta}\) is \(\mu_{2}^{s}\)-invariant.
}

\begin{conjecture}
\label{Brcat-conjecture}
\rev{As it is argued in Section 4.6 of} \cite{Guo:2021aqj},  the category $\mathrm{MF}(Y_{\zeta},W_{\zeta})$ is equivalent to the derived category of sheaves of \rev{$\mathcal{A}_{0}\sharp \mu^{s}_{2}$}-modules over $X$:
\rev{
\begin{eqnarray}
\mathrm{MF}(Y_{\zeta},W_{\zeta})\cong D(X,\mathcal{A}_{0}\sharp \mu^{s}_{2})
\end{eqnarray}
where $\mathcal{A}_{0}$ is the endomorphism algebra of the matrix factorization $\mathbf{T}_{0}$ (see equation (\ref{T0mf})):
\begin{eqnarray}\label{endalgebra}
\mathcal{A}_{0}\sharp \mu^{s}_{2}\cong \mathrm{End}(\oplus_{g\in \mu^{s}_{2}}\mathbf{T}_{0}^{(g)}),
\end{eqnarray}
where $\mathbf{T}_{0}^{(g)}$, denotes $\mathbf{T}_{0}$ but with its representation $\rho_{M}$ twisted by the irreducible representation of  $\mu^{s}_{2}$ labeled by $g$ (for a precise description, see Section 4.5 of  \cite{Guo:2021aqj}, where the case $\mu_{d}$ was analyzed) and $\sharp$  denotes the smash product. In the case at hand, where the 
superpotential is quadratic in the fiber coordinates,  $\mathcal{A}_{0}$ becomes a Clifford algebra $Cl$ (an argument based in quantum field theory is given in Section 4.3.1 of \cite{Guo:2021aqj}) generated by the symbols $\psi_{(i,j)}$, 
$\psi_{(k,0)}$ satisfying
\begin{eqnarray}
\{\psi_{A},\psi_{B}\}=\frac{\partial^{2}W}{\partial z_{A}\partial z_{B}},\qquad A,B \in \{(i,j),(k,0)\}
\end{eqnarray}
therefore $D(X,\mathcal{A}_{0}\sharp \mu^{s}_{2})\cong D(X,Cl\sharp \mu^{s}_{2})$ and in the special case $s=1$:
\begin{eqnarray}
D(X,\mathcal{A}_{0}\sharp \mu_{2})\cong D(X,Cl_{0})
\end{eqnarray}
where $Cl_{0}$ denotes the even part of the sheaf of Clifford algebras 
$\mathcal{A}_{0}$. Then, we conjecture there is an embedding of categories
\begin{eqnarray}
D(Y)\hookrightarrow D(X,Cl\sharp \mu^{s}_{2})
\end{eqnarray}
where \(Y\) is the double cover CY and
the category $D(Y):=D^{b}Coh(Y)$ can be defined following \cite{Kawamata2003}. By homological mirror symmetry we expect in addition that an appropriate definition of the Fukaya category exists for $Y^{\vee}$ and it is related to $D(Y)$. However we are not aware of such definitions of the Fukaya category for a singular cyclic cover such as $Y^{\vee}$.
}
\end{conjecture}

\begin{remark} In \cite{borisov2018clifford} a construction for so-called double mirrors involving noncommutative resolutions is given. From the point of view of the GLSM, this implies that most, if not all, the GLSMs we propose in this section, must have a geometric phase. This is not immediately clear from our construction (and not required for our results). We expect to return to this in a sequel.
\end{remark}

\subsubsection{The \texorpdfstring{\(A\)}{A}-periods of abelian GLSMs}
We can apply the discussion in \S\ref{subsec:abelian}
to the present situation.
To write down the \(A\)-periods,
the only missing piece of information is
an admissible contour \(L_{t}\subset\mathfrak{t}_{\mathbb{C}}\).
To this end, we will need the results in \S\ref{subsec:toric-git}.

Given \(t\in\mathfrak{t}_{\mathbb{C}}\),
let \(\operatorname{Re}(t)=\zeta\in\mathbb{R}^{s}\) as before.
For \(\zeta\) regular, it determines the geometry of the GLSM via
the symplectic quotient (\ref{sympquotglsm})
\begin{equation*}
Y_{\zeta} = \mu^{-1}(\zeta)\slash \mathsf{G}.
\end{equation*}
We may assume that \(\zeta\) belongs to a cone
of maximal dimension in the secondary fan \(S\Sigma\).
In general, the fan \(S\Sigma\) may be singular (even non-simplicial).
Denote by \(S\Sigma'\) a simplicialization\footnote{It is important to remark that, when fixing the GLSM data, we are automatically choosing a subdivision \(S\Sigma'\) i.e.~the secondary fan describing GLSM phases has implicitly chosen a simplicialization.} of \(S\Sigma\).
We may as well assume that \(\zeta\) belongs to
the interior of a cone \(\tau\) of maximal dimension in \(S\Sigma'\).
According to our construction, \(S\Sigma\) and \(S\Sigma'\) are
fans in the Euclidean space \(\mathfrak{t}\).

\begin{definition}
\label{def:a-periods}
The \(A\)-periods of a non-commutative resolution are defined to be
\begin{equation}
\label{eq:a-periods}
Z_{\mathfrak{B}}(q_{1},\ldots,q_{s})=\int_{L} F(\sigma_{1},\ldots,\sigma_{s})
f_{\mathfrak{B}}(\sigma_{1},\ldots,\sigma_{s})
q_{1}^{-\mathrm{i}\sigma_{1}}\cdots q_{s}^{-\mathrm{i}\sigma_{s}}\mathrm{d}\sigma,
\end{equation}
where
\begin{equation*}
\begin{split}
F(\sigma_{1},\ldots,\sigma_{s}) = \prod_{i=1}^{r}\prod_{j=1}^{m_{i}}
\Gamma \left(2\mathrm{i}\sum_{m=1}^{s} \sigma_{m}\theta_{i,j}^{m}\right)
&\prod_{i=1}^{r}\prod_{j=1}^{m_{i}}
\Gamma \left(-\mathrm{i}\sum_{m=1}^{s} \sigma_{m}\theta_{i,j}^{m}+\frac{1}{2}\right)\\
&\times\prod_{i=1}^{r}\Gamma\left(-\mathrm{i}\sum_{m=1}^{s}
\sigma_{m}\theta_{i,0}^{m}+\frac{1}{2}\right)
\end{split}
\end{equation*}
is the product of Gamma functions in \eqref{eq:partition-fcn-abelian},
\(f_{\mathfrak{B}}(\sigma)\) is a brane factor, and we 
give an explicit construction for
\(L\) in Appendix \ref{app:contours}.
\end{definition}
\begin{remark}
\label{rmk:poles-of-f}
The function \(F(\sigma_{1},\ldots,\sigma_{s})\) has a pole 
at \((\sigma_{1},\ldots,\sigma_{s})\)
whenever 
\begin{align*}
\sum_{m=1}^{s} \sigma_{m}\theta_{i,j}^{m}\in \frac{\mathrm{i}\mathbb{Z}_{\geqslant 0}}{2},~
\sum_{m=1}^{s} \sigma_{m}\theta_{i,j}^{m}\in -\frac{\mathrm{i}}{2}+\mathrm{i}\mathbb{Z}_{\leqslant 0},
~\mbox{or}~
\sum_{m=1}^{s} \sigma_{m}\theta_{i,0}^{m}\in -\frac{\mathrm{i}}{2}+\mathrm{i}\mathbb{Z}_{\leqslant 0}.
\end{align*}
\end{remark}

For simplicity, we introduce
\begin{equation}
\label{eq:def-q-p}
Q_{i,j}(\sigma):=\sum_{m=1}^{s}\sigma_{m}\theta_{i,j}^{m},~\mbox{and}~
P_{i}(\sigma):=\sum_{m=1}^{s}\sigma_{m}\theta_{i,0}^{m}=\sum_{j=1}^{m_{i}} Q_{i,j}(\sigma).
\end{equation}
The function \(F(\sigma)\) is then transformed into
\begin{equation*}
F(\sigma) = \prod_{i=1}^{r}\prod_{j=1}^{m_{i}}
\Gamma \left(2\mathrm{i} Q_{i,j}(\sigma)\right)\prod_{i=1}^{r}\prod_{j=1}^{m_{i}}
\Gamma \left(-\mathrm{i}Q_{i,j}(\sigma)+\frac{1}{2}\right)
\prod_{i=1}^{r}\Gamma\left(-\mathrm{i}P_{i}(\sigma)+\frac{1}{2}\right).
\end{equation*}

Using the identities
\begin{equation*}
\Gamma(z)\Gamma(1-z)=\frac{\pi}{\sin \pi z},~\mbox{and}~
\Gamma(z)\Gamma\left(z+\frac{1}{2}\right) = 2^{1-2z} \sqrt{\pi}\Gamma(2z),
\end{equation*}
the formula for \(F(\sigma)\) can be simplified into
\begin{equation}
\label{eq:period-abelian-simplified}
\rev{
F(\sigma)=\frac{\sqrt{\pi}^{p}\pi^{r}}{2^{p}}\prod_{i=1}^{r}\prod_{j=1}^{m_{i}}
\frac{2^{2\mathrm{i}Q_{i,j}(\sigma)}
\Gamma\left(\mathrm{i}Q_{i,j}(\sigma)\right)}{\cos(\mathrm{i}\pi Q_{i,j}(\sigma))}
\prod_{i=1}^{r}\frac{1}{\cos(\mathrm{i}\pi P_{i}(\sigma))}
\Gamma\left(\mathrm{i}P_{i}(\sigma)+\frac{1}{2}\right)^{-1}
}.
\end{equation}
Now let us focus on the brane factors.
Let us begin with a baby example.
\begin{instance}[Calabi--Yau
double cover of \(\mathbf{P}^{1}\)]
\label{ex:cy-double-over-p1}
Consider the GLSM data associated to 
a Calabi--Yau double cover of \(\mathbf{P}^{1}\).
\begin{itemize}
\item \(V=\mathbb{C}^{6}=\mathbb{C}^{2}\times\mathbb{C}^{4}\) with coordinates 
\((\phi_{1,1},\phi_{2,1},z_{1,1},z_{2,1},z_{1,0},z_{2,0})\);
\item \(G=\mathbb{C}^{\ast}\) such that 
\(\phi_{i,j}\) has weight \(2\) and \(z_{i,j}\) has weight \(-1\).
\item The \(R\)-weight of \(\phi_{i,j}\) is \(0\) whereas
the \(R\)-weight of \(z_{i,j}\) is \(1\);
\item \(W\) is the superpotential for singular double cover
\begin{equation*}
W(\phi,z) = \sum_{i=1}^{2} z_{i,1}^{2}\phi_{i,1} +  
\sum_{k=1}^{2} z_{k,0}^{2}f_{k}(\phi),
\end{equation*}
where \(f_{1}(\phi),f_{2}(\phi)\) are the
defining (linear) equations for the branch locus.
It is also clear that \(W\) has \(R\)-weight \(2\).
\end{itemize}
Consider square matrices \(\eta_{1,0},\eta_{1,1},\eta_{2,0},\eta_{2,1}\) and 
\(\bar{\eta}_{1,0},\bar{\eta}_{1,1},\bar{\eta}_{2,0},\bar{\eta}_{2,1}\) satisfying the Clifford relations 
\begin{equation*}
\{\eta_{i,j},\eta_{k,l}\}=\{\bar{\eta}_{i,j},\bar{\eta}_{k,l}\}=0~\mbox{and}~
\{\eta_{i,j},\bar{\eta}_{k,l}\}=\delta_{ik}\delta_{jl}.
\end{equation*}
One can construct the matrices as follows. Consider the 
exterior algebra \(\wedge^{\bullet}\mathbb{C}^{4}\);
this is a complex vector space of dimension \(16\). 
Denote by \(\{e_{1,0},e_{1,1},e_{2,0},e_{2,1}\}\) the standard basis of \(\mathbb{C}^{4}\). 
(This label is more convenient.) Let 
\begin{equation}
\eta_{i,j}:=\iota_{e_{i,j}}\colon \wedge^{\bullet}\mathbb{C}^{4}\to 
\wedge^{\bullet}\mathbb{C}^{4}~
\mbox{and}~\bar{\eta}_{k,l}:=e_{k,l}\wedge -\colon 
\wedge^{\bullet}\mathbb{C}^{4}\to \wedge^{\bullet}\mathbb{C}^{4}.
\end{equation}
It is easy to check that \(\eta_{i,j}\) and \(\bar{\eta}_{k,l}\) 
obey the Clifford or anticommutation relations.

For convenience, let us denote by \(\mathcal{I}\)
the index set \(\{(i,j)~|~1\le i\le 2,~0\le j\le 1\}\). 
Let \(\mathbb{C}v\) be a one dimensional trivial \(G\) and \(R\) representation.

Set
\begin{equation*}
\mathrm{M}=\operatorname{Span}_{\mathbb{C}}\left\{\prod_{(i,j)\in I}
\bar{\eta}_{i,j}v\Bigm\vert\mbox{\(I\) is a subset of \(\mathcal{I}\)}\right\}
\end{equation*}
and 
\begin{equation*}
\mathbf{T}_{0}:=\sum_{i=1}^{2}z_{i,1}\eta_{i,1}+
\sum_{i=1}^{2}z_{i,1}\phi_{i,1}\bar{\eta}_{i,1}+
\sum_{k=1}^{2}z_{k,0}\eta_{k,0}+
\sum_{k=1}^{2}z_{k,0}f_{k}(\phi)\bar{\eta}_{k,0}.
\end{equation*}
From the commutator relations among \(\eta_{i,j}\) and \(\bar{\eta}_{k,l}\),
we have \(\mathbf{T}_{0}^{2} = W\cdot\mathrm{id}_{\mathrm{M}}\).
We also require the factorization \(\mathbf{T}_{0}\)
to be \(G\) and \(R\)-equivariant; namely
\begin{align}
\rho_{\mathrm{M}}(g)^{-1} \mathbf{T}_{0}(\rho(g)\cdot (\phi,z)) 
\rho_{\mathrm{M}}(g) &= \mathbf{T}_{0}(\phi,z)\label{ex-11}\\
R_{\mathrm{M}}(\lambda) \mathbf{T}_{0}(R(\lambda)\cdot (\phi,z)) 
R_{\mathrm{M}}(\lambda)^{-1} &= \lambda\mathbf{T}_{0}(\phi,z)\label{ex-21}.
\end{align}
The equation \eqref{ex-11} is transformed into
\begin{align*}
&\sum_{i=1}^{2} g^{-1}z_{i,1}\rho_{\mathrm{M}}(g)^{-1}\eta_{i,1}\rho_{\mathrm{M}}(g) + 
\sum_{i=1}^{2} g z_{i,1}\phi_{i,1}\rho_{\mathrm{M}}(g)^{-1}\bar{\eta}_{i,1}\rho_{\mathrm{M}}(g)\\
&+\sum_{k=1}^{2} g^{-1}z_{k,0}\rho_{\mathrm{M}}(g)^{-1}\eta_{k,0}\rho_{\mathrm{M}}(g) + 
\sum_{k=1}^{2} g z_{k,0} f_{k}(\phi)\rho_{\mathrm{M}}(g)^{-1}\bar{\eta}_{k,0}\rho_{\mathrm{M}}(g) = \mathbf{T}_{0}(\phi,z)
\end{align*}
which yields
\begin{equation*}
\rho_{\mathrm{M}}(g)^{-1}\eta_{i,j}\rho_{\mathrm{M}}(g) = g\eta_{i,j}~\mbox{and}~
\rho_{\mathrm{M}}(g)^{-1}\bar{\eta}_{i,j}\rho_{\mathrm{M}}(g) = g^{-1}\bar{\eta}_{i,j}.
\end{equation*}
Similarly, the equation \eqref{ex-21} is transformed into
\begin{align*}
&\sum_{i=1}^{2} \lambda z_{i,1}R_{\mathrm{M}}(\lambda)\eta_{i,1}R_{\mathrm{M}}(\lambda)^{-1} + 
\sum_{i=1}^{2} \lambda z_{i,1}\phi_{i}R_{\mathrm{M}}(\lambda)\bar{\eta}_{i,1}R_{\mathrm{M}}(\lambda)\\
&+\sum_{k=1}^{2} \lambda z_{k,0}R_{\mathrm{M}}(\lambda)\eta_{k,0}R_{\mathrm{M}}(\lambda)^{-1}+ 
\sum_{k=1}^{2} \lambda z_{k,0} f_{k}(\phi)R_{\mathrm{M}}(\lambda)\bar{\eta}_{k,0}R_{\mathrm{M}}(\lambda)^{-1} = 
\lambda\mathbf{T}_{0}(\phi,z)
\end{align*}
which yields
\begin{equation*}
R_{\mathrm{M}}(\lambda)\eta_{i,j}R_{\mathrm{M}}(\lambda)^{-1} = \eta_{i,j}~\mbox{and}~
R_{\mathrm{M}}(\lambda)\bar{\eta}_{i,j}R_{\mathrm{M}}(\lambda)^{-1} = \bar{\eta}_{i,j}.
\end{equation*}
From the relations above, we have
\begin{align*}
\rho_{\mathrm{M}}(g) \bar{\eta}_{I} v &= 
\rho_{\mathrm{M}}(g) \bar{\eta}_{I} \rho_{\mathrm{M}}(g)^{-1}
\rho_{\mathrm{M}}(g) v =  g^{|I|} \bar{\eta}_{I}v,\\
R_{\mathrm{M}}(\lambda) \bar{\eta}_{I} v &= 
R_{\mathrm{M}}(\lambda) \bar{\eta}_{I} R_{\mathrm{M}}(\lambda)^{-1}
R_{\mathrm{M}}(\lambda)  v = \bar{\eta}_{I} v.
\end{align*}
Let us now compute the trace \(\operatorname{Tr}(R_{\mathrm{M}}(e^{\mathrm{i}\pi}) 
\rho_{\mathrm{M}}(e^{2\pi\sigma}))\).
Note that \(R_{\mathrm{M}}\) acts trivially on \(\mathrm{M}\) and therefore it suffices
to compute \(\operatorname{Tr}(\rho_{\mathrm{M}}(e^{2\pi\sigma}))\). It follows that
\begin{equation*}
\operatorname{Tr}(\rho_{\mathrm{M}}(e^{2\pi\sigma})) = 
\sum_{I} \genfrac(){0pt}{0}{4}{|I|} e^{2\pi |I|\sigma} = (1+e^{2\pi \sigma})^{4}.
\end{equation*}
We also remark that if \(\mathbb{C}v\) is chosen to have \(G\)-weight \(q\) and 
\(R\)-weight \(r\), i.e.,
\begin{equation*}
\rho_{\mathrm{M}}(g) v = g^{q} v,~\mbox{and}~R_{\mathrm{M}}(\lambda)v = \lambda^{r} v,
\end{equation*}
then the trace becomes \(e^{2\pi q\sigma} e^{\mathrm{i}\pi r} (1+e^{2\pi \sigma})^{4}\);
this is another way to generate linearly independent \(A\)-periods but we will use a different choice of generators hence we can keep the weight of  \(\mathbb{C}v\) fixed to be zero.

Consider another matrix factorization
\begin{align*}
\begin{split}
\mathbf{T}_{1}=z_{1,1}^{2}\eta_{1,1}&+\phi_{1,1}\bar{\eta}_{1,1}+
z_{2,1}\eta_{2,1}+
z_{2,1}\phi_{2,1}\bar{\eta}_{2,1}\\&+
\sum_{k=1}^{2}z_{k,0}\eta_{k,0}+
\sum_{k=1}^{2}z_{k,0}f_{k}(\phi)\bar{\eta}_{k,0}.
\end{split}
\end{align*}
The \(G\)- and \(R\)-equivariant conditions for \(\mathbf{T}_{1}\) now imply
for \((i,j)\ne (1,1)\)
\begin{align*}
\rho_{\mathrm{M}}(g)^{-1}\eta_{i,j}\rho_{\mathrm{M}}(g) = g\eta_{i,j}~&\mbox{and}~
\rho_{\mathrm{M}}(g)^{-1}\bar{\eta}_{i,j}\rho_{\mathrm{M}}(g) = g^{-1}\bar{\eta}_{i,j}\\
R_{\mathrm{M}}(\lambda)\eta_{i,j}R_{\mathrm{M}}(\lambda)^{-1} = \eta_{i,j}~&\mbox{and}~
R_{\mathrm{M}}(\lambda)\bar{\eta}_{i,j}R_{\mathrm{M}}(\lambda)^{-1} = \bar{\eta}_{i,j}.
\end{align*}
and 
\begin{align*}
\rho_{\mathrm{M}}(g)^{-1}\eta_{1,1}\rho_{\mathrm{M}}(g) = g^{2}\eta_{1,1}~&\mbox{and}~
\rho_{\mathrm{M}}(g)^{-1}\bar{\eta}_{1,1}\rho_{\mathrm{M}}(g) = g^{-2}\bar{\eta}_{1,1}\\
R_{\mathrm{M}}(\lambda)\eta_{1,1}R_{\mathrm{M}}(\lambda)^{-1} = 
\lambda\eta_{1,1}~&\mbox{and}~
R_{\mathrm{M}}(\lambda)\bar{\eta}_{1,1}R_{\mathrm{M}}(\lambda)^{-1} = 
\lambda^{-1}\bar{\eta}_{1,1}.
\end{align*}
Let us compute the trace. Again take a trivial representation \(\mathbb{C}v\). We have
for a subset \(I\) with \((1,1)\notin I\)
\begin{equation*}
\begin{cases}
\rho_{\mathrm{M}}(g) \bar{\eta}_{I} v &= g^{|I|} \bar{\eta}_{I}v\\
R_{\mathrm{M}}(\lambda) \bar{\eta}_{I}v &= \bar{\eta}_{I}v
\end{cases}
\end{equation*}
and for a subset \(I\) with \((1,1)\in I\)
\begin{equation*}
\begin{cases}
\rho_{\mathrm{M}}(g) \bar{\eta}_{I} v &= g^{|I|+1} \bar{\eta}_{I}v\\
R_{\mathrm{M}}(\lambda) \bar{\eta}_{I} v &= \lambda^{-1}\bar{\eta}_{I}v
\end{cases}
\end{equation*}
Then it follows that 
\begin{equation*}
\operatorname{Tr}(R_{\mathrm{M}}(e^{\mathrm{i}\pi}) 
\rho_{\mathrm{M}}(e^{2\pi\sigma}))=(1+e^{2\pi \sigma})^{3} ( 1 - e^{4\pi\sigma})
=(1+e^{2\pi\sigma})^{4}(1-e^{2\pi\sigma}).
\end{equation*}
We shall emphasize that the squared term \(z_{1,1}^{2}\) produces an extra
factor \((1-e^{2\pi\sigma})\) which provides zeroes killing some
of the poles in the (product) of Gamma functions in the integrand 
as below.

We can carry out the period integrals. Recall that 
the \(A\)-periods are of the form
\begin{equation}
\label{eq:poles}
\int_{L} \Gamma(2\sigma\mathrm{i})^{2}\Gamma\left(-\sigma\mathrm{i}+\frac{1}{2}\right)^{4} 
q^{-\mathrm{i}\sigma}f_{\mathcal{B}}(\sigma)\mathrm{d}\sigma.
\end{equation}
In the present case, 
\begin{equation*}
f_{\mathcal{B}}(\sigma) = (1+e^{2\pi \sigma})^{4}(1-e^{2\pi \sigma})^{m},~m=0,1.
\end{equation*}
\hfill\qedsymbol
\end{instance}

Now we return to the general setup.
Consider a GLSM data 
\begin{eqnarray}
    (V,\rho\colon G\to\mathrm{GL}(V),
R\colon \mathbb{C}^{\ast}\to \mathrm{GL}(V),W)
\end{eqnarray}
associated to a Calabi--Yau double cover, i.e.,
\begin{itemize}
  \itemsep=-3pt
\item \(V=\mathbb{C}^{2p+r}=\mathbb{C}^{p}\times\mathbb{C}^{p+r}\) is a vector space
with coordinates \(\phi_{k,l}\) (\(1\le k\le r\) and \(1\le l\le m_{k}\)) and
\(z_{i,j}\) (\(1\le i\le r\) and \(0\le j\le m_{i}\)).
The coordinates on \(V\) will be abbreviated as \((\phi,z)\);
\item \(G=(\mathbb{C}^{\ast})^{s}\) is an algebraic torus 
defined in \eqref{eq:basic-ses-group} acting on \(V\) via
\begin{equation*}
\begin{split}
(g_{1},\ldots,g_{s})\cdot (\phi,z)
=\left(\left(\prod_{m=1}^{s}g_{m}^{2\theta^{m}_{k,l}}\right)\phi_{k,l},
\left(\prod_{m=1}^{s}g_{m}^{-\theta^{m}_{i,j}}\right)z_{i,j}\right);
\end{split}
\end{equation*}
\item \(R\colon \mathbb{C}^{\ast}\to \mathrm{GL}(V)\) acting on \(\phi_{k,l}\)
with weight zero and on \(z_{i,j}\) with weight one;
\item \(W\) is the superpotential
\begin{equation}
\label{eq:superpot-recall}
W = \sum_{i=1}^{r}\sum_{j=1}^{m_{i}} 
z_{i,j}^{2}\phi_{i,j} + \sum_{k=1}^{r} z_{k,0}^{2} f_{k}(\phi).
\end{equation}
\end{itemize}
The dimension of the Calabi--Yau double cover \(Y\)
from the GLSM data
is \(p-s\) and the number of FI-parameters is \(s\). 
Notice that \(s\) is the number of FI-parameters for the base 
toric variety.
However, based on Proposition \ref{prop:hodge-numbers},
these are all FI-parameters for the double cover as well.

We construct a matrix factorization
\(\mathfrak{B}_{0}:=(\mathrm{M},\rho_{\mathrm{M}},R_{\mathrm{M}},\mathbf{T}_{0})\) 
of \(W\) in the following way. For simplicity, we set
\begin{align*}
\mathcal{I}&:=\{(i,j)~|~1\le i\le r,~0\le j\le m_{i}\},\\
\mathcal{J}&:=\{(i,j)~|~1\le i\le r,~1\le j\le m_{i}\}.
\end{align*}
Note that \(|\mathcal{J}|=p\) and \(|\mathcal{I}|=p+r\).

Let \(\eta_{i,j}\) and \(\bar{\eta}_{i,j}\) with \((i,j)\in\mathcal{I}\)
be matrices satisfying
\begin{equation*}
\begin{split}
\{\eta_{i,j},\eta_{p,q}\}=
\{\bar{\eta}_{i,j},\bar{\eta}_{p,q}\}=0,~\mbox{and}~
\{\eta_{i,j},\bar{\eta}_{p,q}\}=\delta_{ip}\delta_{jq}.
\end{split}
\end{equation*}
Again one can produce these matrices via the exterior algebra.
Consider
\begin{align}\label{T0mf}
\mathbf{T}_{0}:=&\sum_{i=1}^{r}\sum_{j=1}^{m_{i}} z_{i,j}\eta_{i,j} + 
\sum_{i=1}^{r}\sum_{j=1}^{m_{i}} z_{i,j}\phi_{i,j}\bar{\eta}_{i,j}+
\sum_{i=1}^{r} z_{i,0}\eta_{i,0} + 
\sum_{i=1}^{r} z_{i,0}f_{i}(\phi)\bar{\eta}_{i,0}\nonumber\\
=&\sum_{(i,j)\in\mathcal{J}} z_{i,j}\eta_{i,j} + 
\sum_{(i,j)\in\mathcal{J}} z_{i,j}\phi_{i,j}\bar{\eta}_{i,j}+
\sum_{i=1}^{r} z_{i,0}\eta_{i,0} + 
\sum_{i=1}^{r} z_{i,0}f_{i}(\phi)\bar{\eta}_{i,0}.
\end{align}
From the commutator relation, we see that \(\mathbf{T}_{0}^{2}=W\cdot 
\mathrm{id}_{\mathrm{M}}\).
We also require that \(\mathbf{T}_{0}\) to be both \(G\)- and \(\mathbb{C}^{\ast}\)-equivariant
under \(\rho\) and \(R\) as in Example \ref{ex:cy-double-over-p1}.
A similar calculation yields
\begin{align*}
f_{\mathfrak{B}_{0}}(\sigma)&=\operatorname{Tr}(R_{\mathrm{M}}(e^{\mathrm{i}\pi})
\rho_{\mathrm{M}}(e^{2\pi\sigma_{1}}\cdots e^{2\pi\sigma_{s}}))\\
&=\operatorname{Tr}(\rho_{\mathrm{M}}(e^{2\pi\sigma_{1}}\cdots e^{2\pi\sigma_{s}}))\\
&=\prod_{(i,j)\in\mathcal{J}}(1+e^{2\pi\sum_{m=1}^{s}\sigma_{m}\theta^{m}_{i,j}})
\prod_{k=1}^{r}(1+e^{2\pi\sum_{m=1}^{s}\sigma_{m}\theta^{m}_{k,0}})\\
&=\prod_{(i,j)\in\mathcal{I}} (1+e^{2\pi\sum_{m=1}^{s}\sigma_{m}\theta^{m}_{i,j}}).
\end{align*}
In particular, the brane factor \(f_{\mathfrak{B}_{0}}(\sigma)\)
provides zeroes at \(\sigma\) when
\begin{align*}
\sum_{m=1}^{s}\sigma_{m}\theta^{m}_{i,j}\in \frac{\mathrm{i}}{2}+\mathrm{i}\mathbb{Z}
\end{align*}
for some \((i,j)\in\mathcal{I}\).

We can construct various matrix factorizations from
subsets of \(\mathcal{J}\). Let 
\(J\subset\mathcal{J}\) be a subset. Consider
\begin{align}
\begin{split}
\mathbf{T}_{J}:=&
\sum_{(i,j)\in J} z_{i,j}^{2}\eta_{i,j} + 
\sum_{(i,j)\in J} \phi_{i,j}\bar{\eta}_{i,j}+
\sum_{(i,j)\notin J} z_{i,j}\eta_{i,j} + 
\sum_{(i,j)\notin J} z_{i,j}\phi_{i,j}\bar{\eta}_{i,j}\\
&+\sum_{i=1}^{r} z_{i,0}\eta_{i,0} + 
\sum_{i=1}^{r} z_{i,0}f_{i}(\phi)\bar{\eta}_{i,0}.
\end{split}
\end{align}
\begin{remark}
    It is worth noticing that in the construction of
    the matrix factorization \(\mathbf{T}_{J}\), we 
    did not use the full information from the toric data of 
    our base toric variety \(X\). Instead, we only use the subset \(J\) of
    \(1\)-cones representing the intersection of hyperplanes in \(X\).
\end{remark}
When \(J=\emptyset\), we recover \(\mathbf{T}_{0}\).
Now we can compute the brane factor for \(\mathbf{T}_{J}\).
The equivariance conditions 
\begin{align}
\rho_{\mathrm{M}}(g)^{-1} \mathbf{T}_{J}(\rho(g)\cdot (\phi,z)) 
\rho_{\mathrm{M}}(g) &= \mathbf{T}_{J}(\phi,z)\label{eq:g-equiv-general}\\
R_{\mathrm{M}}(\lambda) \mathbf{T}_{J}(R(\lambda)\cdot (\phi,z)) 
R_{\mathrm{M}}(\lambda)^{-1} &= \lambda\mathbf{T}_{J}(\phi,z)\label{eq:r-equiv-general}
\end{align}
imply that 
\begin{align*}
&\begin{cases}
\rho_{\mathrm{M}}(g)^{-1}\eta_{i,j}\rho_{\mathrm{M}}(g) = 
(\prod_{m=1}^{s} g_{m}^{\theta^{m}_{i,j}})\eta_{i,j}\\
\rho_{\mathrm{M}}(g)^{-1}\bar{\eta}_{i,j}\rho_{\mathrm{M}}(g) = 
(\prod_{m=1}^{s} g_{m}^{-\theta^{m}_{i,j}})\bar{\eta}_{i,j}\\
R_{\mathrm{M}}(\lambda)^{-1}\eta_{i,j}R_{\mathrm{M}}(\lambda) = \eta_{i,j}\\
R_{\mathrm{M}}(\lambda)^{-1}\bar{\eta}_{i,j}R_{\mathrm{M}}(\lambda) = \bar{\eta}_{i,j}
\end{cases}
~\mbox{for}~(i,j)\notin J
\end{align*}
and that
\begin{align*}
&\begin{cases}
\rho_{\mathrm{M}}(g)^{-1}\eta_{i,j}\rho_{\mathrm{M}}(g) = 
(\prod_{m=1}^{s} g_{m}^{2\theta^{m}_{i,j}})\eta_{i,j}\\
\rho_{\mathrm{M}}(g)^{-1}\bar{\eta}_{i,j}\rho_{\mathrm{M}}(g) = 
(\prod_{m=1}^{s} g_{m}^{-2\theta^{m}_{i,j}})\bar{\eta}_{i,j}\\
R_{\mathrm{M}}(\lambda)^{-1}\eta_{i,j}R_{\mathrm{M}}(\lambda) = \lambda\eta_{i,j}\\
R_{\mathrm{M}}(\lambda)^{-1}\bar{\eta}_{i,j}R_{\mathrm{M}}(\lambda) = 
\lambda^{-1}\bar{\eta}_{i,j}
\end{cases}
~\mbox{for}~(i,j)\in J.
\end{align*}
Let \(\mathbb{C}v\) be the trivial \(G\)- and \(\mathbb{C}^{\ast}\)-representation.
In the present case, \(\mathrm{M}\) is spanned by 
\((\prod_{J\subset \mathcal{J}} \bar{\eta}_{J}) v\) where
the notation \(\bar{\eta}_{J}\) means \(\prod_{(i,j)\in J} \bar{\eta}_{i,j}\)
and we use the lexicographic ordering to define this product.
We have shown
\begin{proposition}
\label{prop:brane-factor}
For a subset \(J\subset\mathcal{J}\),
the brane factor \(f_{\mathfrak{B}_{J}}(\sigma)\) associated to
the matrix factorization \(\mathfrak{B}_{J}=(\mathrm{M},\rho_{\mathrm{M}},R_{\mathrm{M}},
\mathbf{T}_{J})\) is given by
\begin{equation}
f_{\mathfrak{B}_{J}}(\sigma)=\prod_{(i,j)\in\mathcal{I}} (1+e^{2\pi\sum_{m=1}^{s}\sigma_{m}\theta^{m}_{i,j}})
\prod_{(i,j)\in J} (1-e^{2\pi\sum_{m=1}^{s}\sigma_{m}\theta^{m}_{i,j}}).
\end{equation}
where $J$ can include the empty set.
\end{proposition}

Combined this with \eqref{eq:period-abelian-simplified}, we have the 
following proposition.
\begin{proposition}
\label{prop:simplified-integrand}
For any subset \(J\subset\mathcal{J}\), we have
\rev{
\begin{align*}
&F(\sigma) f_{\mathfrak{B}_{J}}(\sigma)\\
&=2^{r}\sqrt{\pi}^{p+2r}
\prod_{i=1}^{r}\prod_{j=1}^{m_{i}}
\left(2^{2\mathrm{i}Q_{i,j}(\sigma)}
e^{2\pi Q_{i,j}(\sigma)}\Gamma\left(\mathrm{i}Q_{i,j}(\sigma)\right)\right)\frac{\displaystyle\prod_{(i,j)\in J} (1-e^{2\pi Q_{i,j}(\sigma)})}
{\displaystyle\prod_{i=1}^{r}\Gamma\left(\mathrm{i}P_{i}(\sigma)+\frac{1}{2}\right)}
\end{align*}
}
where \(Q_{i,j}(\sigma)\) and \(P_{i}(\sigma)\)
are defined in \eqref{eq:def-q-p}.
\end{proposition}
\begin{proof}
Recall that 
\begin{equation*}
\cos z = \frac{1}{2}(e^{\mathrm{i}z}\rev{+}e^{-\mathrm{i}z}).
\end{equation*}
We have
\begin{align*}
\displaystyle\cos (\mathrm{i}\pi Q_{i,j}(\sigma)) =& 
\displaystyle\frac{1}{2}(e^{-\pi Q_{i,j}(\sigma)}+e^{\pi Q_{i,j}(\sigma)})
=\displaystyle\frac{e^{-\pi Q_{i,j}(\sigma)}}{2}(1+e^{2\pi Q_{i,j}(\sigma)})\\
\cos (\mathrm{i}\pi P_{i}(\sigma)) =& 
\displaystyle\frac{1}{2}(e^{-\pi P_{i}(\sigma)}+e^{\pi P_{i}(\sigma)})
=\displaystyle\frac{e^{-\pi P_{i}(\sigma)}}{2}(1+e^{2\pi P_{i}(\sigma)}).
\end{align*}
By Proposition \ref{prop:brane-factor}, we have
\begin{equation*}
f_{\mathfrak{B}_{J}}(\sigma)=\prod_{i=1}^{r}\prod_{j=1}^{m_{i}}
(1+e^{2\pi Q_{i,j}(\sigma)})\prod_{i=1}^{r}(1+e^{2\pi P_{i}(\sigma)})
\prod_{(i,j)\in J} (1-e^{2\pi Q_{i,j}(\sigma)}).
\end{equation*}
Using the fact that \(P_{i}(\sigma)=\sum_{j=1}^{m_{i}}Q_{i,j}(\sigma)\), we have
\begin{align*}
&F(\sigma) f_{\mathfrak{B}_{J}}(\sigma)\\
&=\sqrt{\pi}^{p+3r}
\prod_{i=1}^{r}\prod_{j=1}^{m_{i}}
\left(2^{\mathrm{i}Q_{i,j}(\sigma)}\Gamma\left(\mathrm{i}Q_{i,j}(\sigma)\right)\right)
\frac{\displaystyle\prod_{(i,j)\in J} (1-e^{2\pi Q_{i,j}(\sigma)})}
{\displaystyle\prod_{i=1}^{r}\Gamma\left(\mathrm{i}P_{i}(\sigma)+\frac{1}{2}\right)}
\end{align*}
as desired.
\end{proof}

\subsection{\texorpdfstring{\(A\)}{A}-periods and GKZ systems}
To support mirror symmetry for singular Calabi--Yau double covers, 
in this subsection, we will prove that the Picard--Fuchs system for
the \(B\)-periods for \(\mathcal{Y}^{\vee}\to V\) governs the \(A\)-periods
for the category of matrix factorizations with 
an appropriate curved superpotential introduced in \eqref{eq:superpot}. We 
recall that in the present situation the Picard--Fuchs system is the same as
the GKZ system \cite{Lee-period-integrals} and hence we will only manipulate
the GKZ system.


To relate our \(A\)-periods with 
the GKZ systems, the key point is the following change of variables
\begin{equation}
\label{eq:change-of-variable}
q_{k} :=\prod_{i=1}^{r}\prod_{j=1}^{m_{i}} \left(-4x_{i,j}\right)^{\theta_{i,j}^{k}}
\prod_{i=1}^{r}x_{i,0}^{-\theta^{k}_{i,0}}.
\end{equation}

\begin{definition}
\label{def:a-periods-change-of-variables}
Let \(\mathbf{x}=(x_{i,j})_{(i,j)\in\mathcal{I}}\) and set
\begin{equation*}
\hat{Z}_{\mathfrak{B}}(\mathbf{x}):=
\frac{1}{(\prod_{i=1}^{r} x_{i,0})^{1/2}}Z_{\mathfrak{B}}(q_{1},\ldots,q_{s}).
\end{equation*}
\end{definition}

\begin{theorem}
\label{thm:main-theorem}
Under the change of variables \eqref{eq:change-of-variable}, 
the functions \(\hat{Z}_{\mathfrak{B}_{J}}(\mathbf{x})\) satisfy
the GKZ system \(\mathcal{M}_{A}^{\beta}\)
with \(A\) and \(\beta\) defined in \S\ref{subsec:gkz}\textrm{(7)}.
\end{theorem}

Once the change of variables is discovered, the proof follows from a calculation 
with manipulating identities between Gamma functions straightforwardly.

Let us begin with an observation. There is an isomorphism 
\(\operatorname{ker}(A)\cong \operatorname{ker}(B)\)
by removing all the \((i,0)\)\textsuperscript{th} component.
The integral vector \(\theta^{k}=(\theta^{k}_{i,j})\in\operatorname{ker}(B)\) 
defined in \S\ref{subsubsec:a-period-glsm} admits
a unique lifting to \(\operatorname{ker}(A)\). 
By abuse of notation, such a lifting is denoted by \(\theta^{k}\). Notice that 
the \((i,0)\)\textsuperscript{th} component of \(\theta^{k}\)
is given by \(-\theta_{i,0}^{k}\).

\begin{proof}[Proof of Theorem \ref{thm:main-theorem}]
From the change of variables, it is clear
that \(\hat{Z}_{\mathfrak{B}_{J}}(\mathbf{x})\) is killed by 
the Euler operators in \(\mathcal{M}_{A}^{\beta}\).
All we have to check is that the box operators
annihilate \(\hat{Z}_{\mathfrak{B}_{J}}(\mathbf{x})\).

Pick \(\ell\in\operatorname{ker}(A)\) and write \(\ell =
\sum_{k=1}^{s} n_{k} \theta^{k}=\ell^{+}-\ell^{-}\).
Let \(\ell^{+}=(\alpha_{i,j})\) and \(\ell^{-}=(\beta_{i,j})\).
We have \(\ell = (\alpha_{i,j}-\beta_{i,j})\) and 
\begin{align*}
&\prod_{i=1}^{r}\prod_{j=1}^{m_{i}} \left(-4x_{i,j}\right)^{\ell_{i,j}}
\prod_{i=1}^{r}x_{i,0}^{\ell_{i,0}}\\
&=\left(\prod_{i=1}^{r}\prod_{j=1}^{m_{i}} \left(-4x_{i,j}
\right)^{\alpha_{i,j}}
\prod_{i=1}^{r}x_{i,0}^{\alpha_{i,0}}\right)
\left(\prod_{i=1}^{r}\prod_{j=1}^{m_{i}} \left(-4x_{i,j}\right)^{-\beta_{i,j}}
\prod_{i=1}^{r}x_{i,0}^{-\beta_{i,0}}\right).
\end{align*}
By a direct computation, we see that (\(\hat{Z}\equiv \hat{Z}_{\mathfrak{B}_{J}}(\mathbf{x})\))
\begin{align*}
&\left(\prod_{i=1}^{r}\prod_{j=1}^{m_{i}} \left(-4x_{i,j}\right)^{\alpha_{i,j}}
\prod_{i=1}^{r}x_{i,0}^{\alpha_{i,0}}\right)
\partial^{\ell^{+}} \hat{Z}\\ 
&=\prod_{i=1}^{r} \prod_{j=1}^{m_{i}}\prod_{m=1}^{\alpha_{i,j}} 
(-4)\left(-\mathrm{i}\sum_{k=1}^{s}\theta_{i,j}^{k}\sigma_{k}-m+1\right)
\cdot \prod_{i=1}^{r}\prod_{m=1}^{\alpha_{i,0}} 
\left(\mathrm{i}\sum_{k=1}^{s}\theta_{i,0}^{k}\sigma_{k}-m+\frac{1}{2}\right)\hat{Z}.
\end{align*}
Likewise, we have
\begin{align*}
&\left(\prod_{i=1}^{r}\prod_{j=1}^{m_{i}} \left(-4x_{i,j}\right)^{\beta_{i,j}}
\prod_{i=1}^{r}x_{i,0}^{\beta_{i,0}}\right)
\partial^{\ell^{-}} \hat{Z}\\ 
&=\prod_{i=1}^{r} \prod_{j=1}^{m_{i}}\prod_{m=1}^{\beta_{i,j}} 
(-4)\left(-\mathrm{i}\sum_{k=1}^{s}\theta_{i,j}^{k}\sigma_{k}-m+1\right)
\cdot \prod_{i=1}^{r}\prod_{m=1}^{\beta_{i,0}} 
\left(\mathrm{i}\sum_{k=1}^{s}\theta_{i,0}^{k}\sigma_{k}-m+\frac{1}{2}\right)\hat{Z}.
\end{align*}
It follows that 
\begin{align*}
&\left(\prod_{i=1}^{r}\prod_{j=1}^{m_{i}} \left(-4x_{i,j}\right)^{\alpha_{i,j}}
\prod_{i=1}^{r}x_{i,0}^{\alpha_{i,0}}\right)\partial^{\ell^{-}} \hat{Z}\\
&=\left(\prod_{i=1}^{r}\prod_{j=1}^{m_{i}} \left(-4x_{i,j}\right)^{\ell_{i,j}}
\prod_{i=1}^{r}x_{i,0}^{\ell_{i,0}}\right)
\cdot\left(\prod_{i=1}^{r}\prod_{j=1}^{m_{i}} \left(-4x_{i,j}\right)^{\beta_{i,j}}
\prod_{i=1}^{r}x_{i,0}^{\beta_{i,0}}\right)\partial^{\ell^{-}} \hat{Z}\\
&=\frac{1}{(\prod_{i=1}^{r} x_{i,0})^{1/2}}
\int_{L} F(\vec{\sigma})f_{\mathfrak{B}_{J}}(\vec{\sigma})
\prod_{i=1}^{r} \prod_{j=1}^{m_{i}}\prod_{m=1}^{\beta_{i,j}} 
(-4)\left(-\mathrm{i}\sum_{k=1}^{s}\theta_{i,j}^{k}\sigma_{k}-m+1\right)\\
&\times\prod_{i=1}^{r}\prod_{m=1}^{\beta_{i,0}}
\left(\mathrm{i}\sum_{k=1}^{s}\theta_{i,0}^{k}\sigma_{k}-m+\frac{1}{2}\right)
\prod_{k=1}^{s}\left(\prod_{i=1}^{r}\prod_{j=1}^{m_{i}} \left(-4x_{i,j}\right)^{\theta_{i,j}^{k}}
\prod_{i=1}^{r}x_{i,0}^{-\theta_{i,0}^{k}}\right)^{-\mathrm{i}\sigma_{k}+n_{k}}\mathrm{d}\sigma.
\end{align*}
Here \(\vec{\sigma}=(\sigma_{1},\ldots,\sigma_{s})\).
Under the change of variables
\begin{equation*}
\tau_{k} = \sigma_{k} + \mathrm{i}n_{k},
\end{equation*}
the last quantity in the above equation is transformed into
\begin{align*}
&\frac{1}{(\prod_{i=1}^{r} x_{i,0})^{1/2}}
\int_{L+\mathrm{i}\vec{n}} F(\vec{\tau}-\mathrm{i}\vec{n})
f_{\mathfrak{B}_{J}}(\vec{\tau}-\mathrm{i}\vec{n})\\
&\times\prod_{i=1}^{r} \prod_{j=1}^{m_{i}}\prod_{m=1}^{\beta_{i,j}} 
(-4)\left(-\mathrm{i}\sum_{k=1}^{s}\theta_{i,j}^{k}(\tau_{k}-\mathrm{i}n_{k})-m+1\right)\\
&\times\prod_{i=1}^{r}\prod_{m=1}^{\beta_{i,0}} 
\left(\mathrm{i}\sum_{k=1}^{s}\theta_{i,0}^{k}(\tau_{k}-\mathrm{i}n_{k})-m+\frac{1}{2}\right)\\
&\times
\prod_{k=1}^{s}\left(\prod_{i=1}^{r}\prod_{j=1}^{m_{i}} 
\left(-4x_{i,j}\right)^{\theta_{i,j}^{k}}
\prod_{i=1}^{r}x_{i,0}^{-\theta_{i,0}^{k}}\right)^{-\mathrm{i}\tau_{k}}\mathrm{d}\tau.
\end{align*}

We claim that
\begin{align*}
\begin{split}
&\prod_{i=1}^{r} \prod_{j=1}^{m_{i}}\prod_{m=1}^{\alpha_{i,j}} 
(-4)\left(-\mathrm{i}\sum_{k=1}^{s}\theta_{i,j}^{k}\sigma_{k}-m+1\right)\\
&\times\prod_{i=1}^{r}\prod_{m=1}^{\alpha_{i,0}} 
\left(\mathrm{i}\sum_{k=1}^{s}\theta_{i,0}^{k}\sigma_{k}-m+\frac{1}{2}\right)
F(\vec{\sigma})\\
=~&\prod_{i=1}^{r} \prod_{j=1}^{m_{i}}\prod_{m=1}^{\beta_{i,j}} 
(-4)\left(-\mathrm{i}\sum_{k=1}^{s}\theta_{i,j}^{k}(\sigma_{k}-\mathrm{i}n_{k})-m+1\right)\\
&\times\prod_{i=1}^{r}\prod_{m=1}^{\beta_{i,0}} 
\left(\mathrm{i}\sum_{k=1}^{s}\theta_{i,0}^{k}(\sigma_{k}-\mathrm{i}n_{k})-m+\frac{1}{2}\right)\\
&\times F(\vec{\sigma}-\mathrm{i}\vec{n}).
\end{split}
\end{align*}
Here the notation \(\vec{\tau}\) and \(\vec{n}\) are the obvious ones.
Grating this equality, 
let us explain how the claim implies the theorem.
From the claim, we have
\begin{align*}
&\frac{1}{(\prod_{i=1}^{r} x_{i,0})^{1/2}}
\int_{L+\mathrm{i}\vec{n}} F(\vec{\tau}-\mathrm{i}\vec{n})
f_{\mathfrak{B}_{J}}(\vec{\tau}-\mathrm{i}\vec{n})\\
&\times\prod_{i=1}^{r} \prod_{j=1}^{m_{i}}\prod_{m=1}^{\beta_{i,j}} 
(-4)\left(-\mathrm{i}\sum_{k=1}^{s}\theta_{i,j}^{k}(\tau_{k}-\mathrm{i}n_{k})-m+1\right)\\
&\times\prod_{i=1}^{r}\prod_{m=1}^{\beta_{i,0}} 
\left(\mathrm{i}\sum_{k=1}^{s}\theta_{i,0}^{k}(\tau_{k}-\mathrm{i}n_{k})-m+\frac{1}{2}\right)\\
&\times\prod_{k=1}^{s}\left(\prod_{i=1}^{r}\prod_{j=1}^{m_{i}} 
\left(-4x_{i,j}\right)^{\theta_{i,j}^{k}}
\prod_{i=1}^{r}x_{i,0}^{-\theta_{i,0}^{k}}\right)^{-\mathrm{i}\tau_{k}}\mathrm{d}\tau\\
=~&\frac{1}{(\prod_{i=1}^{r} x_{i,0})^{1/2}}
\int_{L+\mathrm{i}\vec{n}} F(\vec{\sigma})
f_{\mathfrak{B}_{J}}(\vec{\sigma})\\
&\times\prod_{i=1}^{r} \prod_{j=1}^{m_{i}}\prod_{m=1}^{\alpha_{i,j}} 
(-4)\left(-\mathrm{i}\sum_{k=1}^{s}\theta_{i,j}^{k}\sigma_{k}-m+1\right)\\
&\times \prod_{i=1}^{r}\prod_{m=1}^{\alpha_{i,0}} 
\left(\mathrm{i}\sum_{k=1}^{s}\theta_{i,0}^{k}\sigma_{k}-m+\frac{1}{2}\right)\\
&\times\prod_{k=1}^{s}\left(\prod_{i=1}^{r}\prod_{j=1}^{m_{i}} 
\left(-4x_{i,j}\right)^{\theta_{i,j}^{k}}
\prod_{i=1}^{r}x_{i,0}^{-\theta_{i,0}^{k}}\right)^{-\mathrm{i}\sigma_{k}}\mathrm{d}\sigma.
\end{align*}
We examine the poles in the integrand
\begin{equation}
\label{eq:integrand}
\begin{split}
F(\vec{\sigma})
f_{\mathfrak{B}_{J}}(\vec{\sigma})
&\prod_{i=1}^{r} \prod_{j=1}^{m_{i}}\prod_{m=1}^{\alpha_{i,j}} 
(-4)\left(-\mathrm{i}\sum_{k=1}^{s}\theta_{i,j}^{k}\sigma_{k}-m+1\right)\\
&\prod_{i=1}^{r}\prod_{m=1}^{\alpha_{i,0}} 
\left(\mathrm{i}\sum_{k=1}^{s}\theta_{i,0}^{k}\sigma_{k}-m+\frac{1}{2}\right).
\end{split}
\end{equation}
According to Proposition \ref{prop:simplified-integrand}, 
the function \(F(\vec{\sigma})
f_{\mathfrak{B}_{J}}(\vec{\sigma})\) has a pole at \(\vec{\sigma}\) whenever
\begin{equation}
\mathrm{i}Q_{i,j}(\vec{\sigma}) \in \mathrm{i}\cdot \mathbb{Z}_{\ge 0}~\mbox{for some}~(i,j)\in\mathcal{J}.
\end{equation}
It suffices to show that for every \(0\le\varepsilon\le 1\), the cycle
\(L+\varepsilon\mathrm{i}\vec{n}\) does not meet the poles in \eqref{eq:integrand}. 
To this end, let \(\vec{\delta} = \vec{\sigma} + \varepsilon\mathrm{i}\vec{n}\in L+\varepsilon\mathrm{i}\vec{n}\).
If \(\mathrm{i}Q_{i,j}(\vec{\delta}) = \mathrm{i}Q_{i,j}(\vec{\sigma}) - \varepsilon Q_{i,j}(\vec{n})
\in \mathbb{Z}_{\le 0}\) (particularly \(\mathrm{i}Q_{i,j}(\vec{\sigma})\) is real), then we must have
\(Q_{i,j}(\vec{n}) \ge  0\). Otherwise, 
\begin{align*}
\mathrm{i}Q_{i,j}(\vec{\sigma}) \in \mathbb{Z}_{\le 0} + \varepsilon Q_{i,j}(\vec{n}) \subset \mathbb{R}_{\le 0}
\end{align*}
and \(L\) can not be a continuous deformation of the real locus \(L_{\mathbb{R}}\)
(cf.~the condition (a) in Definition \ref{definition:admissible-contour}). 
Notice that \(Q_{i,j}(\vec{n}) = \ell_{i,j} = \alpha_{i,j}\). Therefore,
\begin{equation*}
-\mathrm{i}Q_{i,j}(\vec{\sigma}) + \alpha_{i,j}\ge
-\mathrm{i}Q_{i,j}(\vec{\delta}) = -\mathrm{i}Q_{i,j}(\vec{\sigma}) + \varepsilon Q_{i,j}(\vec{n})\ge 
-\mathrm{i}Q_{i,j}(\vec{\sigma})
\end{equation*}
and there exists \(1\le m\le \alpha_{i,j}\) such that
\begin{equation*}
-\mathrm{i}Q_{i,j}(\vec{\delta}) = m-1.
\end{equation*}
In other words, the pole will be annihilated by a zero in the product
in \eqref{eq:integrand}.

Now let us prove the claim. 
Fix \((i,j)\in\mathcal{J}\). We will prove that
\begin{align}
\label{eq:claim-1}
\begin{split}
&\Gamma \left(2\mathrm{i}\sum_{k=1}^{s} \theta_{i,j}^{k}\sigma_{k}\right)
\Gamma \left(-\mathrm{i}\sum_{k=1}^{s} \theta_{i,j}^{k}\sigma_{k}+\frac{1}{2}\right)\\
&\times(-1)^{\alpha_{i,j}}\prod_{m=1}^{\alpha_{i,j}} 
(-4)\left(-\mathrm{i}\sum_{k=1}^{s}\theta_{i,j}^{k}\sigma_{k}-m+1\right)\\
=~&\Gamma \left(2\mathrm{i}\sum_{k=1}^{s} \theta_{i,j}^{k}(\sigma_{k}-\mathrm{i}n_{k})\right)
\Gamma \left(-\mathrm{i}\sum_{k=1}^{s} 
\theta_{i,j}^{k}(\sigma_{k}-\mathrm{i}n_{k})+\frac{1}{2}\right)\\
&\times(-1)^{\beta_{i,j}}\prod_{m=1}^{\beta_{i,j}} 
(-4)\left(-\mathrm{i}\sum_{k=1}^{s}\theta_{i,j}^{k}(\sigma_{k}-\mathrm{i}n_{k})-m+1\right).
\end{split}
\end{align}
Recall that \(Q_{i,j}(\vec{n})=\sum_{k=1}^{s} \theta_{i,j}^{k}n_{k}=\ell_{i,j}=\alpha_{i,j}-\beta_{i,j}\).
We then have
\begin{itemize}
\item[(i)] 
\(\begin{aligned}[t]
\frac{\Gamma \left(2\mathrm{i}\sum_{k=1}^{s} \theta_{i,j}^{k}(\sigma_{k}-\mathrm{i}n_{k})\right)}
{\Gamma \left(2\mathrm{i}\sum_{k=1}^{s} \theta_{i,j}^{k}\sigma_{k}\right)}&=
\frac{\Gamma \left(2\mathrm{i}\sum_{k=1}^{s} \theta_{i,j}^{k}\sigma_{k}+
2(\alpha_{i,j}-\beta_{i,j})\right)}
{\Gamma \left(2\mathrm{i}\sum_{k=1}^{s} \theta_{i,j}^{k}\sigma_{k}\right)}\\
&=\frac{\displaystyle\prod_{m=-\infty}^{2(\alpha_{i,j}-\beta_{i,j})}
\left(2\mathrm{i}\sum_{k=1}^{s} \theta_{i,j}^{k}\sigma_{k}+m-1\right)}
{\displaystyle\prod_{m=-\infty}^{0}\left(
2\mathrm{i}\sum_{k=1}^{s} \theta_{i,j}^{k}\sigma_{k}+m-1\right)},
\end{aligned}\)
\item[(ii)]
\(\begin{aligned}[t]
&\frac{\Gamma\left(-\mathrm{i}\sum_{k=1}^{s} \theta_{i,j}^{k}\sigma_{k}-(\alpha_{i,j}-\beta_{i,j})+1/2\right)}
{\Gamma\left(-\mathrm{i}\sum_{k=1}^{s} \theta_{i,j}^{k}\sigma_{k}+1/2\right)}\\
&=\frac{\displaystyle\prod_{m=-\infty}^{-(\alpha_{i,j}-\beta_{i,j})}\left(-\mathrm{i}\sum_{k=1}^{s} 
\theta_{i,j}^{k}\sigma_{k}+1/2+m-1\right)}
{\displaystyle\prod_{m=-\infty}^{0}\left(-\mathrm{i}\sum_{k=1}^{s} \theta_{i,j}^{k}\sigma_{k}+1/2+m-1\right)}\\
&=(-2)^{\alpha_{i,j}-\beta_{i,j}}
\frac{\displaystyle\prod_{m=-\infty}^{-(\alpha_{i,j}-\beta_{i,j})}\left(2\mathrm{i}\sum_{k=1}^{s} 
\theta_{i,j}^{k}\sigma_{k}-1-2m+2\right)}
{\displaystyle\prod_{m=-\infty}^{0}\left(2\mathrm{i}\sum_{k=1}^{s} \theta_{i,j}^{k}\sigma_{k}-1-2m+2\right)}\\
&=(-2)^{\alpha_{i,j}-\beta_{i,j}}
\frac{\displaystyle\prod_{m=-\infty}^{0}\left(2\mathrm{i}\sum_{k=1}^{s} \theta_{i,j}^{k}\sigma_{k}-1+2m\right)}
{\displaystyle\prod_{m=-\infty}^{\alpha_{i,j}-\beta_{i,j}}\left(2\mathrm{i}\sum_{k=1}^{s} 
\theta_{i,j}^{k}\sigma_{k}-1+2m\right)},
\end{aligned}\)
\item[(iii)]
\(\begin{aligned}[t]
&\frac{\displaystyle(-4)^{\beta_{i,j}}\prod_{m=1}^{\beta_{i,j}}
\left(-\mathrm{i}\sum_{k=1}^{s}\theta_{i,j}^{k}(\sigma_{k}-\mathrm{i}n_{k})-m+1\right)}
{\displaystyle(-4)^{\alpha_{i,j}}\prod_{m=1}^{\alpha_{i,j}} 
\left(-\mathrm{i}\sum_{k=1}^{s}\theta_{i,j}^{k}\sigma_{k}-m+1\right)}\\
&=4^{\beta_{i,j}-\alpha_{i,j}}
\frac{\displaystyle\prod_{m=1}^{\beta_{i,j}}\left(\mathrm{i}
\sum_{k=1}^{s}\theta_{i,j}^{k}\sigma_{k}+
(\alpha_{i,j}-\beta_{i,j})+m-1\right)}
{\displaystyle\prod_{m=1}^{\alpha_{i,j}} \left(\mathrm{i}
\sum_{k=1}^{s}\theta_{i,j}^{k}\sigma_{k}+m-1\right)}\\
&=4^{\beta_{i,j}-\alpha_{i,j}}
\frac{\displaystyle\prod_{m=\alpha_{i,j}-\beta_{i,j}+1}^{\alpha_{i,j}}\left(
\mathrm{i}\sum_{k=1}^{s}\theta_{i,j}^{k}\sigma_{k}+m-1\right)}
{\displaystyle\prod_{m=1}^{\alpha_{i,j}} \left(\mathrm{i}
\sum_{k=1}^{s}\theta_{i,j}^{k}\sigma_{k}+m-1\right)}\\
&=4^{\beta_{i,j}-\alpha_{i,j}}\frac{\displaystyle
\prod_{m=-\infty}^{0} \left(\mathrm{i}\sum_{k=1}^{s}\theta_{i,j}^{k}\sigma_{k}+m-1\right)}
{\displaystyle\prod_{m=-\infty}^{\alpha_{i,j}-\beta_{i,j}} 
\left(\mathrm{i}\sum_{k=1}^{s}\theta_{i,j}^{k}\sigma_{k}+m-1\right)}\\
&=2^{\beta_{i,j}-\alpha_{i,j}} 
\frac{\displaystyle\prod_{m=-\infty}^{0} \left(2\mathrm{i}
\sum_{k=1}^{s}\theta_{i,j}^{k}\sigma_{k}+2m-2\right)}
{\displaystyle\prod_{m=-\infty}^{\alpha_{i,j}-\beta_{i,j}} 
\left(2\mathrm{i}\sum_{k=1}^{s}\theta_{i,j}^{k}\sigma_{k}+2m-2\right)}.
\end{aligned}\)
\end{itemize}
Multiplying them together, we get \eqref{eq:claim-1}.
It is also clear that
\begin{align*}
&(-1)^{\alpha_{i,0}}
\Gamma\left(-\mathrm{i}\sum_{k=1}^{s} \theta_{i,0}^{k}\sigma_{k}+\frac{1}{2}\right)
\prod_{m=1}^{\alpha_{i,0}} 
\left(\mathrm{i}\sum_{k=1}^{s}\theta_{i,0}^{k}\sigma_{k}-m+\frac{1}{2}\right)\\
&=(-1)^{\beta_{i,0}}
\Gamma\left(-\mathrm{i}\sum_{k=1}^{s} \theta_{i,0}^{k}
(\sigma_{k}-\mathrm{i}n_{k})+\frac{1}{2}\right)
\prod_{m=1}^{\beta_{i,0}} 
\left(\mathrm{i}\sum_{k=1}^{s}\theta_{i,0}^{k}(\sigma_{k}-\mathrm{i}n_{k})-m+\frac{1}{2}\right).
\end{align*}
Combined with \eqref{eq:claim-1}, the validity of the claim
is reduced to the fact that
\begin{equation*}
(-1)^{\sum_{i,j}\beta_{i,j}} = (-1)^{\sum_{i,j}\alpha_{i,j}}.
\end{equation*}
This holds since \((\alpha_{i,j}-\beta_{i,j})=(\ell_{i,j})\in \mathrm{ker}(A)\)
and therefore \(\sum_{i=1}^{r}\sum_{j=0}^{m_{i}} \ell_{i,j} = 0\).
\end{proof}

\begin{conjecture}
\label{main-conjecture}
When \(J\) runs through all subsets in \(\mathcal{J}\), the \(A\)-periods
\(\hat{Z}_{\mathfrak{B}_{J}}(\mathbf{x})\) generate 
the full solution set of \(\mathcal{M}_{A}^{\beta}\).

Note that each \(J\) determines a subvariety \(D_{J}\) in the base toric manifold \(X\).
Here, \(D_{J}\) is the intersection of the corresponding toric
divisors in \(J\). We expect 
\begin{eqnarray}
    \{\hat{Z}_{\mathfrak{B}_{J}}(\mathbf{x})\bigm\vert J\subset \mathcal{J}~\mbox{such that}~\{D_{J}\}_{J}~
    \mbox{is a basis for \(\mathrm{H}^{\bullet}(X;\mathbb{C})\)}\}
\end{eqnarray}
gives a basis for the solution space to \(\mathcal{M}_{A}^{\beta}\).
\end{conjecture}

To illustrate Conjecture \ref{main-conjecture}, we provide a proof for the case $G=\mathbb{C}^{*}$ and $\theta_{(i,j)}=1$ for all $(i,j)$, i.e.~when the base is $X=\mathbf{P}^{n}$. 

\begin{theorem}
\label{thm:conjecture-pn}
    Conjecture \ref{main-conjecture} holds for \(X=\mathbf{P}^{n}\).
\end{theorem}
\begin{proof}
In the present case, the subsets $J\subset 
\mathcal{J}$ can be labeled simply by positive integers $m=0,\ldots,n$. It is straightforward to show that the hemisphere partition function 
$Z_{\mathfrak{B}_{J}}(q)$ can be written as a sum of residues of the poles at $\sigma\in \mathrm{i}
\mathbb{Z}_{\geq 0}$ when $\Re(t)\gg 1$:
\begin{eqnarray*}
Z_{\mathfrak{B}_{J}}(q)=\sum_{l=0}^{\infty}Z_{J}^{(l)} \qquad \text{at \ 
}\Re(t)\gg 1
\end{eqnarray*}
where
\begin{equation}
\rev{
\begin{aligned}
Z_{J}^{(l)}=&\sqrt{\pi}^{n+1+2r}2^{r}\cdot (-4)^{(n+1)l}q^{l}\\
&\times\int_{\gamma_{0}} 
2^{2\mathrm{i}(n+1)\sigma}q^{-\mathrm{i}\sigma}e^{2(n+1)\pi\sigma}
~\frac{\displaystyle\prod_{k=1}^{r}
\Gamma(\mathrm{i}\theta_{k,0}\sigma-\theta_{k,0}l+\frac{1}{2})^{-1}(1-e^{2\pi \sigma})^{m}} 
{\displaystyle\Gamma(1-\mathrm{i}\sigma+l)^{n+1}\sin(\mathrm{i}\pi \sigma)^{n+1} }~\mathrm{d}\sigma   
\end{aligned}
}
\end{equation}
where $\gamma_{0}$ is a small counterclockwise contour surrounding the origin 
$\sigma=0$. To show completeness, it is enough to compute the residue at 
$l=0$. Write the residue $Z_{J}^{(0)}$ as
\rev{
\begin{eqnarray}\label{eqresidues}
Z_{J}^{(0)}&=&\sqrt{\pi}^{n+1+2r}2^{r}\int_{\gamma_{0}} 
\frac{\mathrm{d}\sigma}{\sigma^{n+1}} 
\tilde{q}^{-\mathrm{i}\sigma}f(\sigma)(1-e^{2\pi \sigma})^{m} 
\frac{\sigma^{n+1}}{\sin(\mathrm{i}\pi \sigma)^{n+1} 
}\nonumber\\
&=&\frac{2\pi 
\mathrm{i}\sqrt{\pi}^{n+1+2r}2^{r}}{n!}\frac{\mathrm{d}^{n}}{\mathrm{d}\sigma^{n}}\left.\left(\frac{f(\sigma)\tilde{
q } ^ { -\mathrm{i}
\sigma}(1-e^{2\pi 
\sigma})^{m} \sigma^{n+1}}
{\sin(\mathrm{i}\pi \sigma)^{n+1} 
}\right)\right|_{\sigma=0}
\end{eqnarray}
}
where we defined
\begin{eqnarray}
f(\sigma):=\frac{\prod_{k=1}^{r}
\Gamma(\mathrm{i}\theta_{k,0}\sigma+\frac{1}{2})^{-1}} 
{\Gamma(1-\mathrm{i}\sigma)^{n+1}},\qquad 
\rev{\tilde{q}=q2^{-2(n+1)}}
\end{eqnarray}
since $f(\sigma)$ is a regular function at 
$\sigma=0$, $Z_{J}^{(0)}$ becomes
\rev{
\begin{eqnarray}\label{eqresidues2}
Z_{J}^{(0)}=m!(2\pi)^{m}\frac{2\pi 
\mathrm{i}\sqrt{\pi}^{n+1+2r}2^{r}}{n!}\frac{\mathrm{d}^{n-m}}{\mathrm{d}\sigma^{n-m}}\left.\left(\frac{
f(\sigma)\tilde {
q } ^ { -\mathrm{i}
\sigma} \sigma^{n+1}}
{\sin(\mathrm{i}\pi \sigma)^{n+1} 
}\right)\right|_{\sigma=0}.
\end{eqnarray}
}The expression (\ref{eqresidues2}) has a single term proportional to 
$(-\mathrm{i}\ln \tilde{q})^{n-m}$ (all the rest are lower order in 
$-\mathrm{i}\ln 
\tilde{q}$) given by
\begin{equation}
\rev{
\begin{aligned}
Z_{J}^{(0)}=&\mathrm{i}^{-n-1}m!(2\pi)^{m}\frac{2\pi 
\mathrm{i}\sqrt{\pi}^{n+1+2r}2^{r}}{n!}(-\mathrm{i}\ln \tilde{q})^{n-m}\\
&\times\frac{1}
{\displaystyle\Gamma\left(\frac{1}{2}\right)^{r}\pi^{n+1}}+
\mathcal{O}((\ln\tilde{q})^{n-m-1})        
\end{aligned}
}.
\end{equation}
This shows that all the $n+1$ functions, labeled by $J$, are \rev{linearly} independent functions of $\ln \tilde{q}$. 
\rev{Moreover, the result in \cite{Lee-period-integrals}
implies that the relevant GKZ system has holonomic rank
\(n+1\). Hence they form a basis for the solution set.}
\end{proof}

\rev{
\begin{remark}
Generalizing our results to double covers over arbitrary 
toric manifolds needs more work. Given a subset 
\(J\subset\mathcal{J}\), we can form a toric subvariety \(D_{J}\)
by taking intersection of all the members in \(J\). 
Suppose that \(\{D_{J_{k}}\}\subset \mathrm{H}^{\bullet}(X;\mathbb{C})\)
is a \(\mathbb{C}\)-basis.
To prove the conjecture, we have to show that the \(A\)-periods (cf.~\eqref{def:a-periods-change-of-variables})
associated to the brane factors from 
the matrix factorization \(\mathbf{T}_{J_{k}}\) form a basis to 
the relevant GKZ system.
In \(\mathbf{P}^{n}\) case, 
we evaluate the \(A\)-periods by computing the residues.
When the base is of higher dimension, it is more subtle to establish the linear independence. 
\end{remark}
}

\appendix
\section{Contours}\label{app:contours}

Denote by $\theta_{\alpha}$, $\alpha=1,\ldots, p$, the weights $\theta_{i,j}$ in \eqref{eq:basic-ses-group}. Then the cone \(C\) generated by \(\theta_{\alpha}\) is strongly convex
because \(X\) is proper \cite{Cox-toric-varieties}*{Lemma 14.3.2}. 
Note that \(\zeta\) is a certain stability condition
corresponding to an ample divisor. There are vectors 
\begin{equation}
    v_{1},\ldots,v_{s}\in\mathbb{R}^{s}
\end{equation}
such that \(C\subset\operatorname{Cone}\{v_{1},\ldots,v_{s}\}\) and
\begin{equation}
    \zeta = \sum_{m=1}^{s} \ell^{m} v_{m}
\end{equation}
with \(\ell^{m}>0\) for every \(m=1,\ldots,s\). With this choice, we can also write
\begin{eqnarray}
    \theta_{\alpha} = \sum_{m=1}^{s} b_{\alpha}^{~m} v_{m}
\end{eqnarray}
with \(b_{\alpha}^{~m}>0\) for all \(\alpha=1,\ldots,p\) and \(m=1,\ldots,s\).

Write \(\sigma_{i} = x_{i}+ \mathrm{i}\, y_{i}\), then we can write the 
pairing between $\mathfrak{t}$ and its dual by
\begin{equation}
    v_{i}(x):= \sum_{m=1}^{s}v_{i}^{~m} x_{m}~\mbox{and}~
    v_{i}(y):= \sum_{m=1}^{s}v_{i}^{~m} y_{m},
\end{equation}
where \(v_{i} = (v_{i}^{~m})_{m=1}^{s}\) with respect to 
the standard basis for \(\mathbb{R}^{s}\).
Under this notation, we have
\begin{equation}\label{notationb}
    \begin{cases}
        \theta_{\alpha}(x)=\sum_{m=1}^{s} b_{\alpha}^{~m} v_{m}(x)
        =\sum_{m=1}^{s} b_{\alpha}^{~m} x_{m}',\\
        \theta_{\alpha}(y)=\sum_{m=1}^{s} b_{\alpha}^{~m} v_{m}(y)
        =\sum_{m=1}^{s} b_{\alpha}^{~m} (y_{m}'+\varepsilon_{m}),
    \end{cases}
\end{equation}
where we defined the new variables 
\begin{equation}
    \begin{cases}
        x_{i}':=v_{i}(x),\\
        y_{i}':=v_{i}(y)-\varepsilon_{i}.
    \end{cases}
\end{equation}
Here \(\varepsilon_{i}\)'s are constants to be determined later.

According to Proposition \ref{prop:simplified-integrand}, in order for the contour $L$ to be a deformation of $L_{\mathbb{R}}$ inside $\mathfrak{t}_{\mathbb{C}}\setminus \mathcal{H}$ (as required by Definition \ref{definition:admissible-contour} (a)) it must avoid the poles of the Gamma factors \(\Gamma(2\mathsf{i}\,\theta_{\alpha}(\sigma))\) in the \(A\)-periods. Both conditions are satisfied if we take $L$ to be a graph
\begin{eqnarray}\label{app:contourL}
L:=\{(x,y(x))\mid  x\in\mathbb{R}^{s}\}\subset \mathfrak{t}_{\mathbb{C}},
\end{eqnarray}
satisfying the `pole avoiding condition'
\begin{equation}
    \operatorname{Re}(\theta_{\alpha}(\sigma))=0~\Longrightarrow~
    \operatorname{Im}(\theta_{\alpha}(\sigma))\notin \frac{1}{2}\mathbb{Z}_{\ge 0}.
\end{equation}
Since the charge vectors \(\theta_{\alpha}\) are real, this condition can
be written as
\begin{equation}\label{poleavoidingcondition1}
    \theta_{\alpha}(x)=0~\Longrightarrow~
    \theta_{\alpha}(y)\notin \frac{1}{2}\mathbb{Z}_{\ge 0}.
\end{equation}
\eqref{poleavoidingcondition1} can be achieved if $\sigma$ satisfies the stronger condition
\begin{equation}
\begin{aligned}
    \theta_{\alpha}(x)&=\sum_{m=1}^{s} b_{\alpha}^{~m} x_{m}'=0~\Longrightarrow\\
    &\theta_{\alpha}(y)=\sum_{m=1}^{s} b_{\alpha}^{~m} (y_{m}'+\varepsilon_{m})
        \in \mathbb{R}_{< 0}.
\end{aligned}
\end{equation}
Define the matrix \(B:=(b_{\alpha}^{~m})\).
Then, we will write the graph \(L\) in \eqref{app:contourL} in terms of a vector-valued function \(F_{\varepsilon}\colon\mathbb{R}^{s}\to\mathbb{R}^{s}\),
by setting \(y'=F_{\varepsilon}(x')\), then the pole avoiding conditions can be satisfied by choosing $F$ such that: \emph{the components of the vector-valued function
\(F(x')+\varepsilon\) are negative along the union of 
the hyperplanes $\{\theta_{\alpha}(x)=0\}_{\alpha=0}^{p}$}, and the asymptotics condition in Definition \ref{definition:admissible-contour} (b) is satisfied by requiring (as we will show below):\emph{
\(\langle \zeta,F_{\varepsilon}(x')\rangle\) goes to infinity as \(\|x'\|^{h}\to\infty\), $h\in\mathbb{Z}_{>1}$ along
each direction in $x'$}.

We will construct an explicit function $F$ satisfying these conditions. Consider \(F\) such that
\begin{equation}
    (Bx')_{i}=0~\Longrightarrow~
    (BF(x'))_{i}=0.
\end{equation}
In other words, in terms of standard basis, this means that 
\begin{equation}
    v\in\operatorname{Ker}(\theta_{\alpha})~\Longrightarrow~
    \tilde{F}(v)\in\operatorname{Ker}(\theta_{\alpha})
\end{equation}
where \(\tilde{F}\) is the function \(F\) in the standard basis.


\subsection{A general construction of \texorpdfstring{\(F\)}{F}}
\label{subsec:construction-general}
We can abstract our situation as follows.
Let $L_1,\ldots,L_p$ be distinct linear hyperplanes in $\mathbb{R}^s$ with \(p\ge s\),
and let $\zeta\in\mathbb{R}^s$.
\begin{itemize}
    \item Let $\mathcal{S}$ be the finite collection
consisting of $\mathbb{R}^s$ and all distinct positive-dimensional subspaces obtained as
intersections of members of $\{L_1,\ldots,L_p\}$.
    \item For $W\in\mathcal{S}$,
let $\pi_W\colon\mathbb{R}^s\to W$ denote the orthogonal projection to $W$.
\end{itemize}  
The existence and construction of the desired function \(F\) can be achieved by the following theorem.

\begin{theorem}
\label{thm:contour}
There exists a continuous map $F\colon\mathbb{R}^s\to\mathbb{R}^s$ satisfying
\begin{equation}
F(L_i)\subseteq L_i~\mbox{for every}~i=1,\ldots,p
\label{eq:hyperplane-preservation}
\end{equation}
and
\begin{equation}
\langle\zeta,F(x)\rangle\longrightarrow+\infty~\mbox{as}~\lVert x\rVert\longrightarrow\infty
\label{eq:coercivity}
\end{equation}
if and only if
\begin{equation}
\pi_W(\zeta)\neq0~\mbox{for every}~W\in\mathcal{S}.
\label{eq:compatibility}
\end{equation}
Whenever condition~\eqref{eq:compatibility} holds, 
the map $F$ may be chosen
to be polynomial.
\end{theorem}

\begin{proof}
We first prove necessity for the requirement $\pi_W(\zeta)\neq0$, for any $W$.  Let $W\in\mathcal{S}$ be a positive-dimensional
intersection of some of the hyperplanes.  If $x\in W$, then
condition~\eqref{eq:hyperplane-preservation} implies that $F(x)\in W$.  If
$\pi_W(\zeta)=0$, then $\zeta\perp W$, and consequently
\begin{equation}
\langle\zeta,F(x)\rangle=0~\mbox{for every}~x\in W.
\label{eq:necessity-vanishing}
\end{equation}
Because $W$ is unbounded, equation~\eqref{eq:necessity-vanishing} contradicts
condition~\eqref{eq:coercivity}.  The case $W=\mathbb{R}^s$ says that
$\zeta\neq0$, which is also necessary.

We now assume condition~\eqref{eq:compatibility}.  Choose a nonzero linear
functional $\ell_i\colon\mathbb{R}^s\to\mathbb{R}$ such that
\begin{equation}
L_i=\ker(\ell_i)
\label{eq:defining-functional}
\end{equation}
for each $i$.  For every $W\in\mathcal{S}$, set
\begin{equation}
p_W(x)=\prod_{\{i:\,W\not\subseteq L_i\}}\ell_i(x)^2,
\label{eq:weight-polynomial}
\end{equation}
where an empty product is understood to be $1$.  Define
\begin{equation}
F(x)=\lVert x\rVert^h
\sum_{W\in\mathcal{S}}p_W(x)\pi_W(\zeta),\qquad h\in\mathbb{Z}_{>1}.
\label{eq:construction}
\end{equation}
This is a polynomial map.

Fix $i$ and let $x\in L_i$.  For a summand in~\eqref{eq:construction}, if
$W\subseteq L_i$, then $\pi_W(\zeta)\in W\subseteq L_i$.  If
$W\not\subseteq L_i$, then $p_W(x)=0$, since the product
in~\eqref{eq:weight-polynomial} contains the factor $\ell_i(x)^2$.  Thus every
summand belongs to $L_i$, proving~\eqref{eq:hyperplane-preservation}.

Using the defining property of orthogonal projection, we obtain
\begin{equation}
\langle\zeta,F(x)\rangle
=\lVert x\rVert^h
\sum_{W\in\mathcal{S}}p_W(x)\lVert\pi_W(\zeta)\rVert^2.
\label{eq:positive-pairing}
\end{equation}
All summands on the right-hand side are nonnegative.  Given $u\neq0$, let
$W_u$ be the intersection of all $L_i$ containing $u$, with the intersection
of the empty family interpreted as $\mathbb{R}^s$.  Then $W_u$ is
positive-dimensional and belongs to $\mathcal{S}$.  Moreover,
\begin{equation}
p_{W_u}(u)>0,
\label{eq:strictly-positive-term}
\end{equation}
because $W_u\not\subseteq L_i$ implies $u\notin L_i$ and hence
$\ell_i(u)\neq0$.  Conditions~\eqref{eq:compatibility}
and~\eqref{eq:strictly-positive-term} show that the sum in
equation~\eqref{eq:positive-pairing} is strictly positive for every
$u\neq0$. (Recall that \(p\ge s\).)

The function on the unit sphere given by
\begin{equation}
u\longmapsto
\sum_{W\in\mathcal{S}}p_W(u)\lVert\pi_W(\zeta)\rVert^2
\label{eq:sphere-function}
\end{equation}
is continuous and therefore it 
has a positive minimum $c>0$.  If $x=ru$ with $\lVert u\rVert=1$
and $r\geq1$, every polynomial $p_W$ is homogeneous of an even nonnegative
degree.  It follows that
\begin{equation}
\sum_{W\in\mathcal{S}}p_W(ru)\lVert\pi_W(\zeta)\rVert^2
\geq
\sum_{W\in\mathcal{S}}p_W(u)\lVert\pi_W(\zeta)\rVert^2
\geq c.
\label{eq:radial-lower-bound}
\end{equation}
Combining equations~\eqref{eq:positive-pairing}
and~\eqref{eq:radial-lower-bound} gives
\begin{equation}
\langle\zeta,F(x)\rangle\geq c\lVert x\rVert^2
\quad\text{whenever }\lVert x\rVert\geq1.
\label{eq:quadratic-coercivity}
\end{equation}
This proves condition~\eqref{eq:coercivity} and completes the construction.
\end{proof}

\begin{remark}
The condition~\eqref{eq:compatibility} says that $\zeta$ 
has a nonzero projection
onto every positive-dimensional intersection of the
hyperplane arrangement.  
Thus the condition holds for every $\zeta$ outside
the finite union of proper linear subspaces
\begin{equation}
\bigcup_{W\in\mathcal{S}}W^{\perp}\subsetneq\mathbb{R}^s.
\label{eq:exceptional-set}
\end{equation}
\end{remark}

\subsection{The construction for the contour \texorpdfstring{\(L\)}{L}}
Now we can apply results in 
\S\ref{subsec:construction-general}.
Let \(L_{\alpha}=\operatorname{Ker}(\theta_{\alpha})\)
where 
\begin{equation}
    \theta_{\alpha}(x'):=\sum_{m=1}^{s} b_{\alpha}^{~m} x_{m}'
\end{equation}
and \(\zeta\) chosen 
general enough such
that \(\pi_{W}(\zeta)\ne 0\) for every \(W\in\mathcal{S}\);
this can be achieved since the K\"{a}hler cone
is full-dimensional.

We choose \(\varepsilon_{m}\)'s small and general such that 
\begin{equation}
    -\frac{1}{2} < \sum_{m=1}^{s} b_{\alpha}^{~m}\varepsilon_{m}<0~\mbox{for every}~\alpha
\end{equation}
and then take \(L\) in \eqref{app:contourL} to be the graph of 
\begin{equation}
       F(x')+\varepsilon,
\end{equation}
where \(F\) is constructed as in Theorem \ref{thm:contour}.

Then the pole avoiding condition is automatically true.
Indeed, writing \(y'=F(x')\), we have
\begin{eqnarray}
    \sum_{m=1}^{s} b_{\alpha}^{~m}(y_{m}'+\varepsilon_{m})=
    \sum_{m=1}^{s} b_{\alpha}^{~m}\varepsilon_{m}\in\left(-\frac{1}{2},0\right)\subset\mathbb{R}_{<0}.
\end{eqnarray}
whenever $x'\in\mathrm{Ker}(\theta_{\alpha})$.

\subsection{Asymptotics}
Next we test that the contour $L$ satisfies the asymptotic conditions required in Definition \ref{definition:admissible-contour} (b). First, note that when $\mathsf{G}$ is abelian of rank $s$, as in our models, then \eqref{twistedbdry} becomes (note the sum is over all $2p+r$ weights)
\begin{equation}
\label{eq:boundary-W}
\widetilde{W}_{\mathrm{eff},q}(\sigma)=-\sum_{m=1}^{s}t^{m}\sigma_{m}+2\pi \mathrm{i}\, q(\sigma)-\sum_{j=1}^{2p+r}\theta_{j}(\sigma)\left(\ln\left(\mathrm{i}\,\theta_{j}(\sigma)\right)\right),
\end{equation}
where $q(\sigma)=q^{m}\sigma_{m}$ and $q^{m}\in\mathbb{Z}$. Then, set $x'=ru$ for some fixed vector $u\in\mathbb{R}^{s}$, $\|u\|=1$ and $r\in\mathbb{R}$. Let us first focus on the term 
\begin{equation}
\ln\left(\mathrm{i}\,\theta_{j}(\sigma)\right)=\ln\left(|\theta_{j}(\sigma)|\right)+\mathrm{i}\,\mathrm{Arg}\left(\mathrm{i}\,\theta_{j}(\sigma)\right),
\end{equation}
and, upon choosing a branch cut for $\ln$, the function $\mathrm{Arg}$ is bounded, for example we can choose it such that $\mathrm{Arg}\in(-\pi,\pi)$. As $r\rightarrow +\infty$ The term $\ln\left(\mathrm{i}\,\theta_{j} (\sigma)\right)$ behaves as
\begin{equation}\label{asympt-lnterm}
\ln\left(\left|\theta_{j}(\sigma)|\right)\sim \ln\left(|B_{j}(y')+B_{j}(\varepsilon)\right|\right)= \ln\left(\left|r^{h}\sum_{W\in\mathcal{S}}p_{W}(ru)B_{j}(\pi_{W}(\zeta))+B_{j}(\varepsilon)\right|\right),
\end{equation}
where we defined $B_{j}(y')=\sum_{m=1}^{s}b_{j}^{m}y'_{m}$, by extending the definition \eqref{notationb} to all weights, the only difference is that some of the $b_{j}^{m}$ coefficients can be negative for $j>p$. The asymptotics at large $r$ of \eqref{asympt-lnterm} can be further simplified to 
\begin{itemize}
    \item $\ln\left(|\theta_{j}(\sigma)|\right)\sim h'\ln r+\ln c_{j}(u,\zeta)$, if $B_{j}(\pi_{W}(\zeta))\neq 0$ for at least one $W$. With some integer $h'>h$ and a positive bounded constant $c_{j}(u,\zeta)$.
    \item $\ln\left(|\theta_{j}(\sigma)|\right)\sim \ln \left(|B_{j}(\varepsilon)|\right)$, if $B_{j}(\pi_{W}(\zeta))=0$ for all $W\in\mathcal{S}$.
    \item Finally we remark, that if we have the case that $B_{j}(ru)=rB_{j}(u)\neq 0$ and  $B_{j}(\pi_{W}(\zeta))=0$ for all $W\in\mathcal{S}$, the asymptotics behave as $\ln\left(|\theta_{j}(\sigma)|\right)\sim \ln r+\ln\left( B_{j}(u)\right)$.
\end{itemize}
Therefore we conclude that all cases, can be summarized as  $\ln\left(|\theta_{j}(\sigma)|\right)\sim f(r)+\ln c_{j}(u)$ where $f(r)$ is independent of $j$ and $c_{j}(u)$ is some positive constant that can depend linearly on $\zeta, u$ or $\varepsilon$ (but we emphasize the dependence on the choose of direction $u$). Then, for a nonanomalous GLSM, we have that $\sum_{j=1}^{2p+r}\theta_{j}=0$ and so, when restricted to $L$ and the direction $x'=ru$, the function \eqref{eq:boundary-W} becomes
\begin{equation}
\label{asymptoticseff}
\widetilde{W}_{\mathrm{eff},q}(\sigma)\sim -\sum_{m=1}^{s}t^{m}\sigma_{m}+2\pi \mathrm{i}\, q(\sigma)-\sum_{j=1}^{2p+r}\theta_{j}(\sigma)\ln c_{j}(u),\qquad \text{as \ }r\rightarrow +\infty
\end{equation}
and therefore
\begin{equation}\label{asymptoticsnc}
\qquad \mathrm{Im}(\widetilde{W}_{\mathrm{eff},q}(\sigma))\sim -\langle \zeta,F(ru)\rangle-\sum_{j=1}^{2p+r}B_{j}(F(ru))\ln c_{j}(u),
\end{equation}
Then, for generic $\zeta$ in the interior of $C$, the dominant term in \eqref{asymptoticsnc} is the first term. Using \eqref{asymptoticseff},
the leading term is nothing but  $-\langle \zeta,F(x')\rangle$, having the desired property
\begin{equation}
    -\langle \zeta,F(ru)\rangle
    \le -cr^{h}
    \longrightarrow - \infty~\mbox{as}~r\to+\infty.
\end{equation}
So $\mathrm{Im}(\widetilde{W}_{\mathrm{eff},q}(\sigma))\rightarrow -\infty$, in an arbitrarily chosen asymptotic direction, as needed. We remark that, as it is evident from \eqref{asymptoticsnc}, for some choices of $\zeta$, depending on the weights $\theta_{j}$, it may be the case that the terms can cancel and the asymptotics analysis is more delicate. These choices of $\zeta$ are not considered generic and are identified with the, possibly `quantum corrected,' boundaries of the K\"ahler cone \cite{Morrison-Plesser}.





\bibliographystyle{amsxport}
\bibliography{reference}
\end{document}